\documentclass[10pt,journal,compsoc]{IEEEtran}

\ifCLASSOPTIONcompsoc
  \usepackage[nocompress]{cite}
\else
  \usepackage{cite}
\fi

\ifCLASSINFOpdf
\else
\fi
\usepackage{amsthm}
\usepackage{tabulary,xcolor}
\usepackage{graphics}
\usepackage{graphicx}
\usepackage[flushleft]{threeparttable}
\usepackage{url}
\usepackage{caption}
\usepackage{adjustbox}
\usepackage{multirow}
\usepackage{multicol}
\usepackage{amsfonts,amsmath,amssymb}
\usepackage{algorithmic}
\usepackage[ruled, lined, linesnumbered, commentsnumbered, longend]{algorithm2e}
\usepackage{subcaption}
\usepackage{csquotes}
\def\enquote#1{`#1'}
\newtheorem{theorem}{Theorem}
\newtheorem{corollary}{Corollary}

\begin{document}
%
\title{Mathematical Analysis of Path MTU Discovery With New Generation Networks}
%
%
%
%

\author{Ishfaq~Hussain and
        Janibul~Bashir
\IEEEcompsocitemizethanks{\IEEEcompsocthanksitem I. Hussain was with the Janibul Bashir's Laboratory, National Institute of Technology, Srinagar, Jammu and Kashmir, 190006, India during this research study. (email: ishfaqhussain90@gmail.com)
\protect \hfil\break
J. Bashir is with the Department
of Information Technology, National Institute of Technology, Srinagar, Jammu and Kashmir, 190006, India. (email: janibbashir@nitsri.ac.in)
}
\thanks{}}

%
%

\markboth{Hussain \MakeLowercase{\textit{et al.}}}%
{A performance analysis}

\IEEEtitleabstractindextext{%
\begin{abstract}
\textit{Path MTU Discovery} (PMTUD) was initially designed for \textit{Internet protocol version 4} (IPv4) to prevent the communication loss due to smaller path MTU. This protocol is then further developed for \textit{Internet protocol version 6} (IPv6) with new set of constraints. In IPv4 network, the PMTUD activates when the packet's \textit{Don't Fragment} (DF) bit is set, while as in IPv6, PMTUD is always running for every packet. In this paper we have presented the effects of path mtu discovery in IPv4 \& IPv6 in mathematical, logical and graphical representation. We try to give a mathematical model to the working of path mtu discovery and calculated its behaviour using a transmission of a packet. We analysed the time consumed to transmit a single packet from source to destination in IPv6 network in the presence of PMTUD and similarly in IPv4 network with DF bit \enquote{1}. Based on our analysis, we concluded that the communication time increases with the varying MTU of the intermediate nodes. Moreover, we formulated the mathematical model to determine the communication delay in a network. Our model shows that the asymptotic lower bound for time taken is \(\Omega(n)\) and the asymptotic upper bound is \(\Theta(n^2)\), using Path MTU Discovery. We have find that the packet drop frequency follows the Bernoulli's trials and we were able to define the success probability of the packet drop frequency $a$ $\forall a\le n$ which shows that the probability is higher for packet drop rate for beginning $2\%$ of the total nodes in the path. We further found that $^{n}C_{a}$ possible number of a-combinations without repetitions that can be formed for a particular number of packet drop frequency. The relation between summation \textit{(acts as a coefficient in the time wastage equation)} of each combination and their frequency resulted in symmetric graph and also mathematical and statistical structures to measure time wastage and its behaviour. This also helps in measuring the possible relative maximum, minimum and average time wastage. We also measured the probability of relative maximum, min and average summation for a given value of packet drop frequency and number of nodes in a path.
\end{abstract}

}

\maketitle

\IEEEdisplaynontitleabstractindextext

\IEEEpeerreviewmaketitle


\section{Introduction}

The \textit{Path MTU Discovery} (PMTUD) protocol was initially designed for Internet protocol version 4 (IPv4) \cite{rfc1191,rfc791} with the aim to discover the minimum path MTU of all the links interface in the arbitrary path from source to destination. The goal is to reduce the packet drop frequency in networks. The protocol works for the packets where the \textit{Don't Fragment} (DF) bit is set or in scenarios where the intermediate nodes are not allowed to fragment the incoming packets. The protocol make use of \textit{Internet Control Message Protocol} (ICMP) messages to inform the source about the minimum MTU of the path from the source to the destination using \textit{Packet Too Big} (PTB) messages ~\cite{rfc792}. 

 The ICMP message encapsulates the size of the last node's forwarding MTU value along the cause of the packet drop with a message as \enquote{fragmentation needed and DF bit set}. On reaching this message to the source node it further fragments the packet and re-transmits it into small chunks of size equal to or lower then the said MTU value informed by the \textit{Internet Control Message Protocol Version 4} (ICMPv4) message. It keeps on repeating the process until it reaches to destination \cite{rfc1191}.

In case of IPv6 network, the intermediate nodes are not allowed to fragment the packets \cite{rfc8200}. This decision was carried out due to the reason that the fragmentation is considered harmful \cite{frag} and has many effects on security and network performance of wired, wireless \cite{wireless} , IoT (Internet of Things) \cite{IOT} and 6LoWPAN \cite{6152926} networks. As a result, PMTUD protocol is always active in such networks. The Path MTU discovery v6 uses ICMPv6 message protocol with (Type 2, Code 0) error message as \enquote{Packet too Big} \cite{rfc4443}. On receiving this ICMPv6 packet \cite{rfc4443}, it contains the next Path MTU value of the problem occurred node and the source node regenerates the same packet of size equal to the informed MTU value in ICMPv6 message and re-transmits it and this process keeps on repeating until the packet is successfully transmitted to destination \cite{rfc8201}.

In both the networks, IPv4 and IPv6, Path MTU Discovery results in increase in time delay if the process of re-transmission keeps on repeating and hence adversely affects the network throughput as it largely depends on the time delay. There are many factors for the re-transmission of a single packet through a network path but most of the common factors are :
\begin{enumerate}
\item Random and decreasing MTU value of the nodes in the path.
\item ICMP message unreachable due to firewall restrictions \cite{firewall}, MTU mismatch \cite{mismatch}, routers are configured to not to send ICMP destination unreachable messages,  software bugs responsible for PMTUD failures \cite{behave} and PMTUD Black holes \cite{BB12} causes the source to continuously send packets without knowing the path MTU after every timeout.
\item The change in path MTU value over time due to the change in routing topology \cite{topology,Ptopology} after every next re-transmission tends to increase the use of PMTUD algorithm which results in time consumption in transmitting a packet \cite{rfc1191}.
\end{enumerate}

From the above discussed factors, the aim of this paper falls in the first and the third factors i.e. to analyse the time consumption due to re-transmission of packet in the network, using Path MTU discovery for IPv6 \& IPv4 networks with DF=1. In this paper, we analyzed the best to the worst case situation that could rise using PMTUDv6 in IPv6. The same study is applied to the IPv4 network with Path MTU Discovery v4 with the DF bit set to $'1'$. 

In this paper the time delay due to PMTUD protocol is expressed by a term \enquote{Time Wastage} and is symbolised as \({^p}T_w\). Thus, \enquote{Time Wastage} is defined as the extra time that has been taken by PMTUD protocol to decrease the MTU to the minimum MTU of the path. \\

\textbf{Contributions:} Let us summarize the contributions of this paper :

\begin{itemize}
\item Initial stages of research of measuring the performance of Path MTU Discovery in IPv4 \& IPv6 network using mathematical models and formulas.
\item The packet transmission using the Path MTU Discovery exhibits the mathematical properties which is used to evaluate the network performance, namely latency.
\item We prove that the asymptotic upper bound and lower bound of time wastage in using Path MTU discovery is \(\Theta(n^2)\) and   \(\Omega(n)\) respectively, in IPv6 and IPv4 network with DF=1.
\item We calculated that the maximum time wastage and the minimum time wastage on sending a single packet in Path MTU discovery is a two degree polynomial equation and one degree polynomial equation respectively, in IPv6 and IPv4 network with DF=1.
\item We present a mathematical models to find out the total time wastage in sending a single packet for different scenarios using PMTUD protocol in IPv6 and IPv4 network with DF=1.
\item We measured the success probability for relative maximum, minimum and average summation for a given value of packet drop frequency and number of nodes in the path.
\item We have calculated the relative maximum, minimum and average time wastage for the given value of packet drop frequency and number of nodes.
\item We give a statistical illustration between summation and their frequencies which help in measuring the time wastage.
\item We further elaborated the mathematical structures using graphical and statistical representations.
\item We give the success probability of each packet drop frequency for a given value of number of nodes in a path using Bernoulli's and binomial distribution.
\end{itemize}

The rest of the treatise is as follows. In Section \ref{sec:overview} presents the literature overview of the related study and Section~\ref{sec:motive} briefly targets the motivations and applications of the study. Section~\ref{sec:delay} discusses a theoretical case study in PMTUD, followed by Section~\ref{sec:analysis} where we present the mathematical modelling of Path MTU discovery. In Section~\ref{sec:analysis_result} is an analysis and results and paper concludes with conclusions and future work in Section~\ref{sec:conc}. 
\section{Literature Overview}
\label{sec:overview}
Thiago Lucas et al. \cite{flowfragment} carried out an analytic study of comparing the IPv4 \& IPv6 network using Path MTU Discovery in UDP protocol. The study concentrates on the effect of the Path MTU Discovery in datagram transmission in both the networks, by allowing and disallowing the fragmentation under UDP protocol. In this study, the jitter and the bandwidth are taken as the parameters to analyse the effectiveness of the two protocols with respect to the size of the datagram. The size of transmitting datagram is varied from 100 KB to 5 MB. The results of the comparison illustrated a great variation in IPv4 network and a minor changes in IPv6 network w.r.t. the two parameters. In IPv4, the jitter doubled and the bandwidth halved in case of no fragmentation as compared to the fragmentation. While in IPv6 network, the jitter and the bandwidth shows minor changes. Additionally, there is a reverse effect as compared to IPv4 network i.e., the jitter decreased and the bandwidth increased on the datagrams with restricted fragmentation as compared to unrestricted fragmentation of datagrams. This study focused on the effect of PMTUD on both the protocols using virtual network environmental simulations.

Matthew Luckie \& Ben Stasiewicz in \cite{behave} measures the Path MTU Discovery behaviour for 50 thousand popular websites by measuring the TCP behaviour after the PTB messages reached the host asking to send smaller packets. The authors concluded that the failure of the PMTUD infers to no response from the host to reduce the packet size or to clear the DF bit for tree consecutive times upon sending the PTB messages by remote TCP connection \cite{behave}. They depicted that the failure rate, as predicted in the previous studies for the IPv4, is much lesser than what is stated. It shows that nearly 80\% of the devices act normally on PTB messages for both IPv4 \& IPv6 \cite{behave}. They identified that most of the failure rate is due to the software bugs rather than filtering by the firewalls. Nearly, 62\% of the PMTUD failure rate can be reduced if focused on correcting these software bugs. This study found out that 11\% of web-servers limit themselves to packets with size no longer then 1380 bytes. On basis of this analysis, they proposed three different strategies to cutoff the failure rate in IPv4. First, the operating systems that refuses to lower the packet size below the threshold limit, should set the DF bit, if the path MTU is lower than the threshold limit \cite{behave}. Second, to debug the middleboxs that alter the TCP MSS size to 1380 bytes to ensure the correct forwarding of the PTB messages by identifying their manufactures \cite{behave}. Third, to make aware the system administrator of the importance of forwarding the PTB messages. While this study doesn't present any mathematical modelling or insight chemistry of the Path MTU Discovery on IPv6 network, rather it focuses on the causes for failure of PMTUD in IPv4 \& IPv6 based on TCP behaviour after PTB message are send to host \cite{behave}.

Christopher A. Kent \& Jeffrey C. Mogul \cite{frag} presented the study that changed the whole perspective of the packet transmission, giving rise to the usage and dependency on Path MTU Discovery by suggesting to set the DF bit to $'1'$. In this study, they found out the misuse of fragmentation by hackers and also a factor responsible for the poor performance and quality of the service. Another retrospection on the similar behaviour \cite{retro}, is published by the same authors stating the  further need of improvement of the PMTUD type techniques for reducing the harm that the fragmentation is anticipating today. Both these studies discusses the harm due to fragmentation and promoting the techniques like PMTUD to lower the harms that the fragmentation can cause. Similar work by F. Gont et al. \cite{rfc8021} suggesting to take out the IPv6 atomic fragmentation on the upcoming revision of the core IPv6 protocol specification by stating that \enquote{Its of no real gain and this kind of functionality is undesirable}. They further stated that, this can result in several security vulnerabilities in IPv6 by misusing the atomic fragmentation and  causes interoperability issues with other protocols relying on it. Nearly 57\% of the web-servers that are tested, fail to generate the IPv6 atomic fragments as per the ICMPv6 messages with a size of MTU lower then the 1280 bytes resulting an interoperability issues between the different protocols in the IPv6 stack \cite{rfc8021}.

Matthew Luckie et al. in \cite{firewall} carried out an analysis of Path MTU Discovery in IPv4 network. In this analysis, they have tried to resolve the failure of the PMTUD in modern network system. This study discusses and presents some important issues related to the cause of the failure of PMTUD in IPv4 network. They have presented a detailed analysis on inferring the exact location of the each failure of PMTUD \cite{firewall}.

K. Lahey \cite{rfc2923} presented a draft in IETF that includes the study of several TCP implementation issues with PMTUD. In this study, they presented three problems, namely black hole detection, stretch ACK due to PMTUD, and determining MSS from PMTU. With the explanation of these issues they have given detailed analysis for resolving them. The main motivation of this study is to improve the current conditions of the internet by improving the efficiency and quality of TCP/IP implementation.

\section{Motivations}\label{sec:motive}
We have gone through all of the literature that came in our way through the Search Engine like google and from Libraries using keywords "MTU or Maximum Transmission Unit" , "Path MTU Discovery or PMTUD", "Mathematical modeling/behaviour/aspects/structures of/in Path MTU discovery", "performance/analysis/evaluation of/in PMTUD" and "Surveys/reviews/study on Path MTU Discovery" and also related keyword but we haven't found any resembled or close resemble of the work which we have presented in this paper and most of them are surveys and a review on the behaviour/performance/analysis of Path MTU Discovery in the internet systems, parameters and the components. It was shocking to know that no study in the past has presented the mathematical modeling of the Path MTU Discovery for sending a packet from a source to destination. This paper presents a novel study of showing how that Path MTU discovery exhibits mathematical structure to explain possible performance and behaviour of latency in PMTUD in IPv4 and Ipv6 network. How this study can be harnessed? And what we can gather from this is all depends on how readers can utilize it but we try to give a possible sketch of its applications below:
\begin{itemize}
    \item Pre-examination of the latency in a network running Path MTU Discovery using mathematical modeling.
    \item Measuring the robustness of similar proposed techniques in terms of latency with the Path MTU Discovery.
    \item Using the mathematical models in designing a better simulators that will mimic the real world performance.
    \item Helps in designing a optimised network topology in order to decrease overall latency by ordering the position of nodes of specific mtu values.
    \item Used to compare and measure the network parameters of the proposed method with state-of-the-art Path MTU Discovery when the implementation of the method is not possible in the node. 
\end{itemize}

This work is based on theoretical experiments using proofs and case study which can be used in verifying the new methods, schemes, and algorithms for testing the robustness of the newly proposed methods/schemes/techniques with the start-of-the-art PMTUD algorithm. This would be best fit when the schemes or methods is impossible for the authors to implement it in the node/router due to vendor's side restrictions and because of no open source platform to test the method/algorithm/scheme with the path mtu discovery to get the experimental results. Therefore this study can be used to compare the network parameters of the proposed methods with path mtu discovery using $Matlab$ or similar mathematical simulators by following the benchmarking methodologies given by Bradner et al. in \cite{bench1, bench2}, like we have used in \cite{HUSSAIN2021} for our new proposed algorithm called Dynamic MTU for reducing the packet drop and time delay associated in path MTU Discovery as defined in its preprint \cite{hussain2020measuring}. We strongly believe that this study can be harnessed to design a better protocols and techniques which can replace Path MTU Discovery by solving and tackling the root problems in the Path MTU Discovery protocol, which we tried to define in this paper.
\pdfoutput=1
\section{Theoretical Study}
\label{sec:delay}
In Path MTU Discovery algorithm, the PTB messages is used in form of ICMP error messages to inform the host about the path mtu of the link but due to many interoperability issues of working of ICMP PTB messages made PMTUD very much unreliable in IPv4 and fails mostly in IPv6 network as per the referred studies. We were in the middle of designing a new start-of-art protocol and find a real need of such models in calculating the robustness and performance of our new designed algorithm to that of the Path MTU discovery, for pre-implementation validation of the method. Please note that in this paper we are only discussing the mathematical modelling and analysis of time delay in Path MTU Discovery, which is the most basic parameter in Internet system to measure the performance and Quality of the network, network parameters, network equipment and protocols. 

Before going into the analysis we have begin with a case study where we designed a theoretical network configuration as shown in Figure~\ref{PMTUD_demo}. This network configuration can be used for both IPv4 and IPv6 packets analysis as in this study we are analysing the effect of PMTUD in time delay, so both of them using PMTUD protocol and no separate network configuration and analysis is needed.
Figure~\ref{PMTUD_demo} shows a network configuration between a source and destination , where the source tries to sends a packet of size 1800 octet to destination and follows different steps to complete the transmission with a path of varying MTUs. We will first measure and analyse the time delays using Path MTU Discovery. The outputs from this case study will be used later in the paper, to further analyse time consumption.
Further in Table \ref{par} gives the meaning of the parameters and symbols that have been used throughout the paper.

 \begin{figure}[!b]
 \centering
 \centerline{\includegraphics[width=0.9\columnwidth]{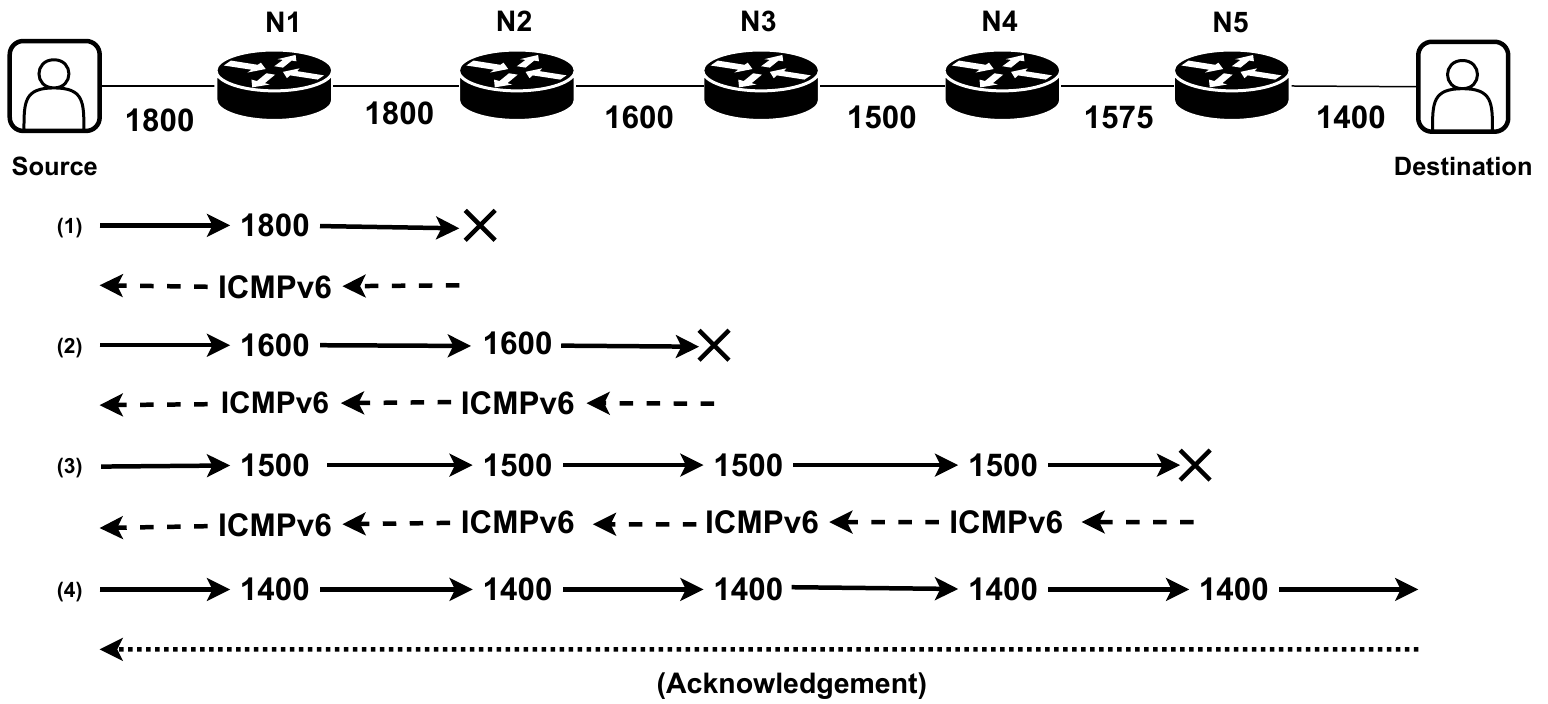}}
 \caption{ Packet transmission in IPv6 network using PMTUD.}
 \label{PMTUD_demo}
 \end{figure}

\begin{table}[!t]
\caption{Definitions\label{par}}
\resizebox{\columnwidth}{!}{%
{\begin{tabular}{c|c}
\hline
Meaning & Symbols \\ [0.5ex]
\hline
Total time wastage & ${^p}T_{w}$ \\
Time wastage at $i^{th}$ re-transmission & $T_{wi}$ \\ 
Average Packet Propagation delay  & ${^p}T_{d}$ \\
Average ICMP Propagation Delay & ${^p}T'_{d}$ \\
Total time delay in PMTUD & ${^p}T$ \\
Time of Fragmentation at source  & $T_f$ \\
Packet drop frequency  & $a$ \\
Number of nodes in a path  & $n$ \\
$i^{th}$ times packet dropped by any node & $n_{[i]}$ \\
Summation Sign  & $\sum$ \\
Sum of n-terms  & $S_n$ \\
Positive Integers  & $\mathbb{Z^+}$ \\
Summation of terms of a combination  & $S_i$ \\
Frequency of summation  & $\nu(S_i)$ \\
A sequence of position of nodes & $(b_{n})_{n=1}^{l}$ \\
A sub-sequence of sequence $(b_{n})_{n=1}^{l}$ & $(b_{n})_{n=1}^{l}$ \\
Distinct number of summations  & $nD(S_i)$ \\
Greater sign  & $>$ \\
Less sign & $<$ \\
Belongs to  & $\in$ \\
Less then or equal  & $\le$ \\
Greater than or equal & $\ge$ \\[1ex] 
 \hline
 \end{tabular}}{}
 }
  \begin{tablenotes}
   \item[1] All of the parameters meaning, applies to a single packet transmission with Path MTU Discovery.
\end{tablenotes}
\end{table}

\subsection{A Case Study with PMTUD}
In Figure~\ref{PMTUD_demo}, the source begins to transmit the  first
transmission with initial packet of size 1800 bytes, the packet travels up-to
node 2, at node 2 the next-interface MTU is lower than the incoming packet size
so the node truncates the packet and sends an ICMPv6 (Type 2) message to source
“Packet too Big”.
It should be noted that, in this study we are assuming the average packet propagation delay of all the links in all of the re-transmission i.e ${^p}T_{d}$ between the links and the average ICMPv6 propagation delay as ${^p}T_{d}$.
The \textit{End-to-End Delay} (E2ED) for the packet to reach at node 2 from source is
$2({^p}T_{d})$, where ${^p}T_{d}$ is average packet propagation delay of the link due to first transmission and the factor $2$ is because it traverses two links. Similarly, time delay for the ICMPv6 message is $2{^p}T'_{d}$, where ${^p}T'_{d}$ is the average ICMPv6 propagation delay for the first transmission. When the source receives the
ICMPv6 message it fragments the packet and initiate the $2^{nd}$ re-transmission. The total time taken for the fragmentation process is $T_f$.
Since, the time wastage up-to $2^{nd}$ re-transmission is :
\begin{equation}
T_{w1} = 2({^p}T_{d} +{^p}T'_{d}) + T_f
\end{equation}

Where \(T_{w1}\) is the time wastage in first transmission.\\
In \(2^{nd}\) re-transmission of packet, the packet travels up-to node 3 and is again truncated by the \(3^{rd}\) node and send the same ICMPv6 (Type 2) message to source \enquote{Packet too Big}. The E2ED from source to the node 3 is \(3{^p}T_{d}\), where the factor 3 is because the packet traveled 3 links. The E2ED for the ICMPv6 message from node 3 to source is \(3{^p}T'_{d}\). At source the packet is again fragmented and initiate the third transmission. The fragmentation at source takes some time to fragment the packet which is \(T_f\). So the E2ED for \(2^{nd}\) re-transmission from source to initiation of \(3^{rd}\) re-transmission is calculated as:
\begin{equation}
T_{w2} = 3({^p}T_{d} +{^p}T'_{d})+T_f
\end{equation}

Similarly, in the \(3^{rd}\) re-transmission the transmission is again failed at node 5 which sends back the ICMPv6 (Type 2) echo message “Packet too Big” to source. Therefore the E2ED for the \(5^{th}\) transmission is the sum of time delay for packet to reach node 5 from source. The time delay of ICMPv6 message from node 5 to source and the fragmentation time at source is:
\begin{equation}
T_{w3} = 5({^p}T_{d} +{^p}T'_{d}) + T_f
\end{equation}

In  \(4^{th}\) re-transmission the packet size is least compared to the packet size of all previous transmissions and is equal to the minimum path mtu of all the links in the path. This transmission gets successful and packet reaches to destination and sends back an \textit{acknowledgement} (ACK) packet to source. Since the E2ED in sending the packet from source to destination in \(4^{th}\) re-transmission is given by:
\begin{align}
T_4 &= {^p}T_{d}(n+1) \\			
T_4 &= {^p}T_{d}(5+1)\\
T_4 &= 6({^p}T_{d})
\end{align}
Where, \(T_4\) is E2ED of \(4^{th}\) transmission which is also a successful transmission and \({^p}T_d\) is the average packet propagation delay between consecutive nodes and \textit{n} is number of nodes in the path.

Therefore, the total loss of time or time wastage \(({^p}T_{w})\) on sending the packet from source to destination from \(1^{st}\) transmission to last transmission using Path MTU Discovery is:
\begin{align}
{^p}T_{w} &= T_{w1} + T_{w2} + T_{w3}\\
{^p}T_{w} &= 2({^p}T_{d} + {^p}T'_{d}) + T_f + 3({^p}T_{d}+{^p}T'_{d})+T_f \notag \\
&+5({^p}T_{d} +{^p}T'_{d}) + T_f \\
{^p}T_{w} &= 10({^p}T_{d} + {^p}T'_{d}) + 3T_f \label{pmtu_waste}
\end{align}
Since, there is no time loss in \(4^{th}\) re-transmission as its successful transmission.

In practical, the propagation delay of all the links in $1^{st}$, $2^{nd}$ and $3^{rd}$ transmissions are not similar as the packet size decreases for every new transmission due to fragmentation the propagation time also decreases. Therefore, the real order of the propagation delay of the links in $1^{st}$, $2^{nd}$, $3^{rd}$ and $4^{th}$ transmissions is:
\begin{align}
    T_{d1}>T_{d2}>T_{d3}>T_{d4} \label{comp}
\end{align}
Where $T_{d1},\ T_{d2},\ T_{d3}\ \&\ T_{d4}$ are propagation delay of the link in $1^{st}$, $2^{nd}$ and $3^{rd}$ transmissions respectively.

However, taking the average packet propagation delay has less effect on accuracy of total time wastage in PMTUD i.e. Equation~\ref{pmtu_waste} and helps in easy mathematical modeling with a very marginal deviation from the real time observation.
Therefore, the total time delay for transmitting the packet in  Path MTU Discovery in this particular case is:
\begin{gather}
Total\ time\ delay =\ T_4 + {^p}T_{w}\notag\\
{^p}T = 6({^p}T_{d}) + 10({^p}T_d +{^p}T'_d) + 3T_f \label{TotalPMTUD}
\end{gather}

Overall summarising, in this section we have taken a case study to calculate the total time wastage in IPv4 and IPv6 network due to the packet drop by smaller path mtu of the path with a limited numbers of intermediate nodes (i.e 5 nodes) in PMTUD. But if we want to calculate the time wastage for any arbitrary nodes which are dropping packet randomly, then the above equations can't be used to calculate the time wastage due to PMTUD in this case. That's why we need to define a equation that will calculate the time delay for $n$ nodes which dropping packet $a$ times in the path. Before going to the calculation of time delay for $n$ nodes, we have defined the effect of time wastage due to PMTUD in the Total time delay for transmitting a packet in PMTUD with an example in the following section.

\subsection{Impact on Latency}
 The time wastage, which we also called as extra time delay or exceed time delay using all these name reflect same meaning, can be understand as a data exchange between Tom and Jerry, where Tom being a transmitter and Jerry being a receiver. The expected total time delay without any truncation of packet is \(T_d\), and when Tom transmits a packet to Jerry and is dropped due to lower path mtu, an ICMP message is send to Tom to reduce the packet size and retransmit it to Jerry. Tom regenerates the packet and re-transmits it to Jerry, which arrives in an expected time delay of \(T_d\). Since the time consumed in the previous failed transmission is named as the time wastage or extra time delay. Now the impact in the total time delay in sending the packet using PMTUD will be :
\[Total\ time\ delay\ ({^p}T)\ =\ T_d\ +\ {^p}T_w\]
Where \({^p}T_d\) is expected time without packet loss and \({^p}T_w\) is time wastage due to failed transmission and \({^p}T\) is actual observed time.
In other words the sum of total time delay of failed transmissions due to successive single packet loss by intermediate nodes is called time wastage or extra time delay.
In the following theorems we have shown different scenarios of time wastage depending upon the order by which nodes dropping a packet and number of nodes involved using PMTUD algorithm. We have draw some further conclusions in the following section on time wastage using PMTUD in more generalised way.

\section{Mathematical Modelling}
\label{sec:analysis}
In this section, we present a mathematical models by theoretical analysis of the time delay or time wastage for $n$ nodes in Path MTU Discovery for both IPv4 \& IPv6 network. We are here again assuming the average packet propagation delay for all the links for any packet transmission as $^{p}T_{d}$ and also taking the average ICMPv6 propagation delay as $^{p}T'_{d}$. We have defined a set of theorems and corollaries and provided them with a complete proof. The mathematical modelling of the time delay varies from best to the worst case of packet drop by the participating nodes. We try to present more generalised and compact models that may cover a wide range of the scenarios.

\subsection{Worst Case Scenario}
The worst case in time wastage in new generation networks when using PMTUD arises at two situations. First, when all the nodes in the path between the source and the destinations are involved in dropping the packet in ascending order. this will further increase in the processing, queuing, fragmentation and re-transmitting delay and hence all of these rise the time wastage. The second situation arises when a packet is dropped by any number of node in the path. 
The time wastage in transmitting a packet in new generation network using path MTU discovery algorithm is measured and analysed at its worst case
\begin{theorem}\label{1}
The maximum or the $worst-case$ total time wastage \(({^p}T_w)\)  for the packet drop using PMTUD algorithm is :
\begin{equation}
    {^p}T_w = S_n({^p}T_d+{^p}T'_d) + nT_f, \\ \label{g_sn}
    where\ n \subset  {Z}^+
\end{equation}  

where \(S_n\) is sum of n-terms and n is defined as the number of nodes between source and destination.
\end{theorem}

\begin{proof}
When the first node drops the packet, an ICMP message is sent to source and the source re transmits a new packet with reduced size, which reaches at node 1. The time wastage till node 1 will be given by the sum of E2ED from source to node 1 and the ICMP message delay from node 1 to source i.e \({^p}T_d +{^p}T'_d\).
Again, if the packet is dropped only by second node, then the time wastage will be \(2({^p}T_p + {^p}T'_d)\) as the node has to reach the source for Ack packet and send the packet again up-to \(2^{nd}\) node which will take \(2({^p}T_p + {^p}T'_d)\) extra time.
Since, if the packet is dropped by any number of nodes in the path between source and destination, then the time wastage will be the factor of effecting node to the sum of E2ED of packet and ICMPv6 packet and by subtracting the two consecutive time wastage given by two consecutive nodes have same difference i.e. \(({^p}T_p + {^p}T'_d)\), which implies it is in form of an arithmetic progression with difference of \(({^p}T_p + {^p}T'_d)\) , so by A.P the \(a_nth\) term is given by 
\begin{equation*}
    a_n = a + (n-1)d ,\  \forall\ n \in  (nodes)\ \to {Z}^+
\end{equation*}
Since, \(a\ ={^p}T_d+{^p}T'_d +T_f \ ,\ d = {^p}T_d+{^p}T'_d \), then the time wastage for incremental nodes dropping packets varies as:
  \begin{gather}
  T_{w1} = {^p}T_d+{^p}T'_d+T_f,\\
 T_{w2} =  2{^p}T_d+2{^p}T'_d+T_f,\\
  \vdots\ \\
 T_{w(n-1)} =  (n-1){^p}T_d+(n-1){^p}T'_d+T_f,\\
 T_{wn} =  n{^p}T_d+n{^p}T'_d+T_f
  \end{gather}
Therefore, the total time wastage is given by adding time wastage of incremental nodes upto $n^{th}-node$ :
\begin{align}
{^p}T_{w}\ =\  &T_{w1} +  T_{w2} + \dots +  T_{w(n-1)} + T_{wn}\\
    {^p}T_{w}\ =\ &({^p}T_d+{^p}T'_d +T_f)+(2{^p}T_d+2{^p}T'_d +T_f) + \\ 
    &\dots +(n{^p}T_d+n{^p}T'_d +T_f)\notag \\
    =\ &({^p}T_d+{^p}T'_d)(1\ +\ 2\ +\ 3\ + \dots + n) + nT_f\notag \\
    =\ &S_n({^p}T_d+{^p}T'_d)+nT_f\notag \\
    {^p}T_{w}\ =\ &S_n({^p}T_d+{^p}T'_d)+nT_f, \ \forall\ n \in {Z}^+
\end{align}
\end{proof}

\begin{corollary}\label{cor1}
The asymptotic upper bound for the $worst-case$ time wastage using Path MTU discovery is of $\Theta(n^2)$.
\end{corollary}
\begin{proof}
Since, from the Theorem \ref{1}, the Equation \ref{g_sn} which can be written as function of n:
\[f(n) = {n^2 \left ({^p}T_d+{^p}T'_d \over 2 \right )} + n \left( {{^p}T_d+{^p}T'_d \over 2}+T_f \right)\]

Which is in the form of two-degree polynomial equation of form \(an^2 + bn + c\ where\ a= [{^p}T_d+{^p}T'_d]/2 ,\ b\ = [{^p}T_d+{^p}T'_d]/2+T_f,\ c= 0\ and\ a>0\), therefore we can present this equation in asymptotic notation which has the positive constant of \(c_1\ ,c_2\ and\ n_0\) therefore the constants can be found as:
\[c_1n^2 \le f(n) \le c_2n^2\]
\[c_1n^2 \le {n^2 \left ({^p}T_d+{^p}T'_d \over 2 \right )} + n \left( {{^p}T_d+{^p}T'_d \over 2}+T_f \right) \le c_2n^2\]
\[c_1 \le {\left ({^p}T_d+{^p}T'_d \over 2 \right )} + \left( {{^p}T_d+{^p}T'_d + 2T_f \over 2n} \right) \le c_2\]

The right hand inequality hold for any vale of $n \to \infty$ by choosing any kind of constant $c_1 \le {^p}T_d+{^p}T'_d$. Similarly, we make the left side inequality hold true for any value of $n \ge 1$ by selecting any constant $c_2 \ge {^p}T_d+{^p}T'_d +T_f$. Since by Choosing $c_1,\ c_2,\ and\ n$ at these given values we can say that:
\[f(n) = \Theta(n^2)\]
\end{proof}

 \begin{figure}[!t]
 \centering
 \centerline{\includegraphics[width=1\columnwidth, height=0.8\columnwidth]{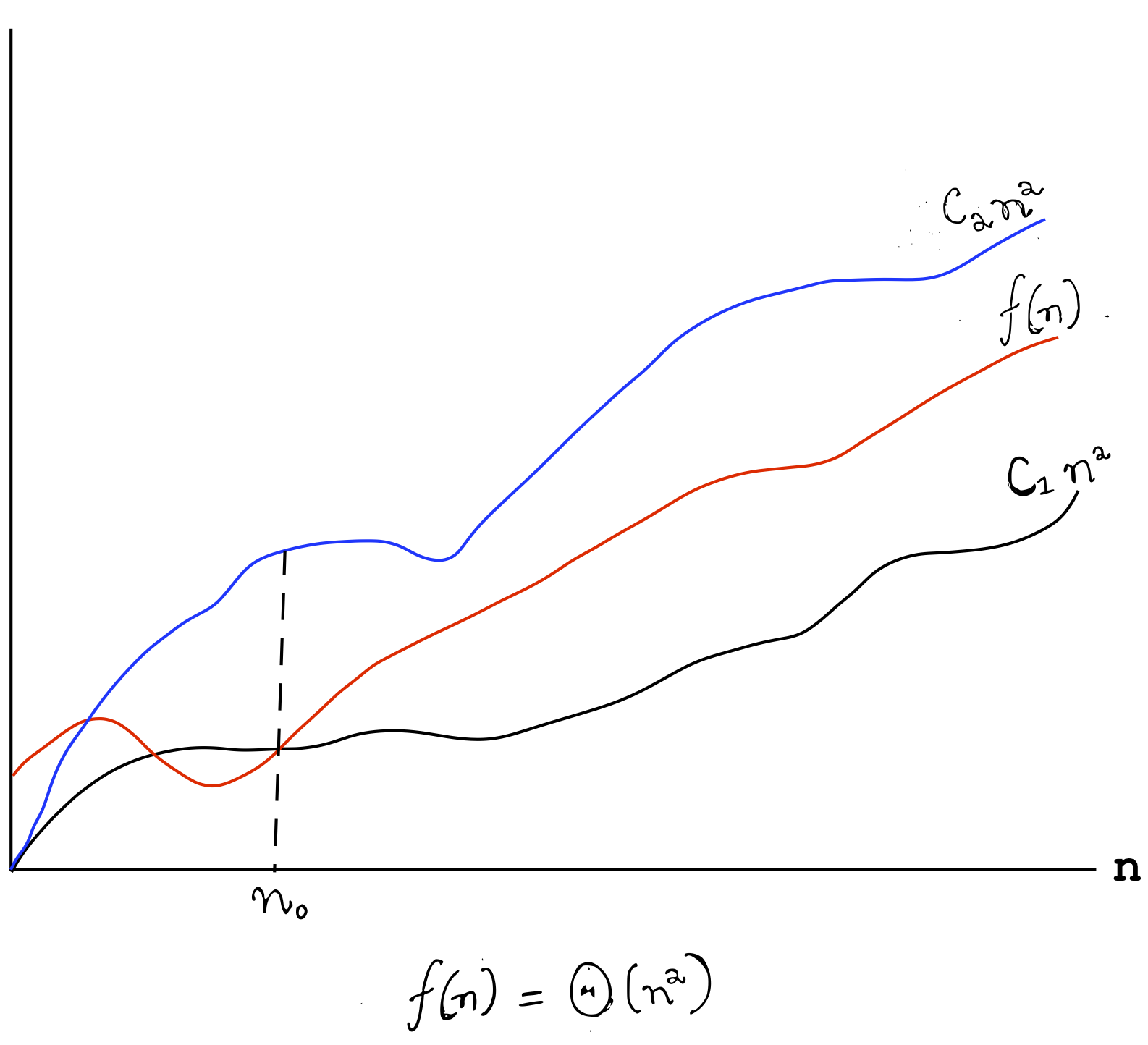}}
 \caption{ Graphical illustration of time complexities in $\Theta$ - notation which gives Upper bound for time function $f(n)$ for positive constants $c_1$, $c_2$ and $n_0$ for all $n \ge n_0$.}
 \label{col1}
 \end{figure}
 
Since, the asymptotic upper bound of $worst-case$ time wastage in PMTUD algorithm is \(\Theta(n^2)\), where $n$ is the number of nodes in the path between source and the destination which the packet has traversed.

Since f(n) has resulted in a asymptotic tight bound of $\Theta(n^2)$ which can yield both upper bound and lower bound of time wastage in PMTUD. The upper bound can also be expressed as $O(n^2)$ \cite{cormen}. However, the lowest bound at worst-case time wastage of PMTUD algorithm is $\Omega(n^2)$ and can't be the overall lower bound of time wastage in PMTUD as in normal packet drop in PMTUD the lower bound of time delay is $\Omega(n)$ as detailed in Corollary \ref{cor2}.

In Figure \ref{col1} is the graphical representation of the Corollary \ref{cor1} i.e $f(n)=\Theta(n^2)$. The line of $c_2(n)^2$ denote the upper bound and $c_1(n)^2$ denotes the lower bound while the f(n) lines between these bounds  $\forall n \ge n_o$, hence $f(n)$ defines the tight bound as $\Theta(n^2)$.
\subsection{Best Case Scenario}
\begin{theorem} \label{theorem_2}
When a packet has been truncate by any single intermediate node using PMTUD algorithm, then the time wastage \(({^p}T_w)\) for transmitting the single packet is given by:
\begin{align}
  {^p}T_w = n_1({^p}T_d + {^p}T'_d) + T_f , \label{g_n1}\\ 
  \forall \  n_1 \subset n\notag
\end{align}
where \(n \in (nodes) \to {Z}^+ ,\) such that \( n_1 \in n_i \) represents the node which dropped the packet.
\end{theorem}
\begin{proof}
Since for any transmission, the average packet propagation delay between hops is \({^p}T_d\) and the average time delay for ICMPv6 message is \({^p}T'_d\), where \({^p}T_d > {^p}T'_d\) as ICMPv6 message can't be greater then 1280 octets. During re-transmission the source node needs to do fragmentation which will take time and is given by \(T_f\).
Lets assume a path where the next node link mtu is lower then the previous node. And if a packet is send with a size greater than the link mtu of the first node then the packet is dropped by all of the nodes in that path and the total time wastage is equal to the sum of all of the time wastage contributed by each node.

Therefore, time wastage of \(1^{st}\) node:
\[{^p}T_d + {^p}T'_d + T_f\]   
Similarly, for other nodes the time wastage is:
\begin{align}
& 1^{st}\ ,  \label{p_eqn}
&&{^p}T_{w}
=1({^p}T_d+{^p}T'_d) +T_f\\
& 2^{nd}\ ,\ 
&&{^p}T_{w}
=2({^p}T_d+{^p}T'_d) +T_f\\	
& 3^{rd}\ ,\ 
&&{^p}T_{w}
=3({^p}T_d+{^p}T'_d) +T_f\\	
& 4^{th}\ ,\ 
&&{^p}T_{w}
=4({^p}T_d+{^p}T'_d) +T_f\\
& 5^{th}\ ,\ 
&&{^p}T_{w}
=5({^p}T_d+{^p}T'_d) +T_f\\	
& 6^{th}\ ,\  \label{q_eqn}
&&{^p}T_{w}
=6({^p}T_d+{^p}T'_d) +T_f
\end{align}

Since, from Equation \ref{p_eqn} to \ref{q_eqn}, we see that the \({^p}T_{w}\) for each node depends on the factor of the position number of the node with \(({^p}T_d+{^p}T'_d)\), having constant throughout \((T_f)\). Therefore, the general equation for the total time wastage for the packet drop by any arbitrary node \(n_1\) is given by:
\begin{align}
{^p}T_{w}\ =\  n_1({^p}T_d+{^p}T'_d) +T_f ,\\
\forall \ n_1 \subset n,\ where\ n \in (nodes) \to {Z}^+.
\end{align}
\end{proof}
\begin{corollary}\label{cor2}
The asymptotic lower bound for the time wastage using Path MTU discovery is of $\Omega(n)$
\end{corollary}
\begin{proof}
From Theorem \ref{theorem_2} the Equation \ref{g_n1} form a linear 1 degree polynomial equation of form \(an + b,\) where \(a\ = ({^p}T_d+{^p}T'_d)\ and\ b = T_f\), which shows that:
\[f(n) = an + b\]
\[f(n) = \Omega(n) \]
$Where,$ \[cn \le {({^p}T_d+{^p}T'_d)n} + T_f\]
Dividing both the sides of inequality by $'n'$ we get:
\[c \le ({^p}T_d+{^p}T'_d) + {T_f\over n}\]
The left hand inequality can be hold for any value n $\ge$ 0 choosing any constant c $\le ({^p}T_d+{^p}T'_d)$

 \begin{figure}[!t]
 \centering
 \centerline{\includegraphics[width=1\columnwidth, height=0.8\columnwidth]{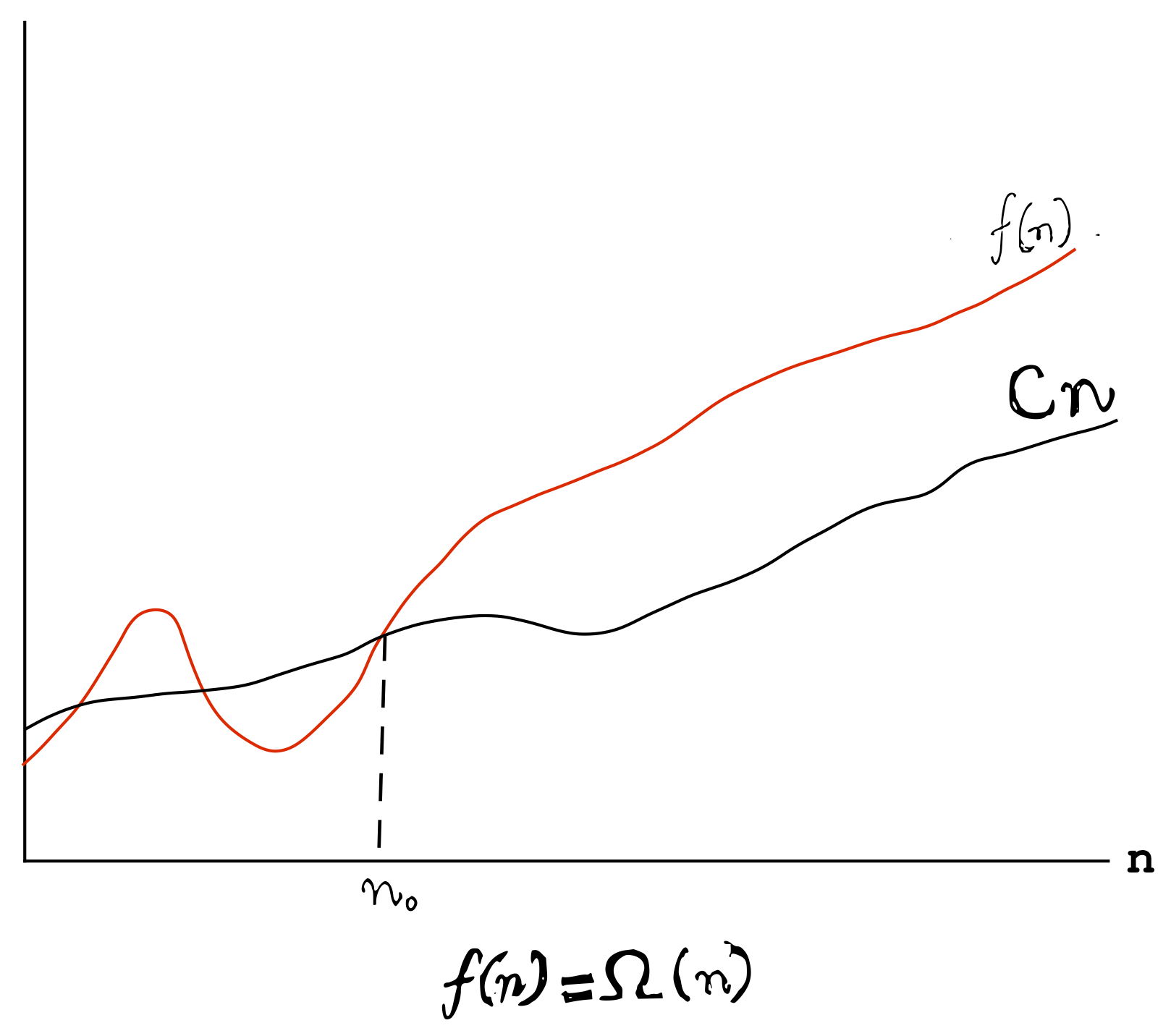}}
 \caption{Graphical illustration of complexities in $\Omega$ - notation which gives lower bound for function $f(n)$ for positive constants $n_0$ and $c$ for all $n \ge n_0$. }
 \label{col2}
 \end{figure}

Therefore, the asymptotic lower bound is $\Omega(n)$. which implies that this is the best case scenario of the time wastage in Path MTU discovery.
\end{proof}

In Figure \ref{col2} is an illustration of the Corollary \ref{cor2} i.e $f(n)=\Omega(n)$ which is the asymptotic lower bound. The line of $c(n)$ denotes the limit of the lower bound \& the line of f(n) defines the lower bound as $\Omega(n)$ $\forall n \ge n_o$. This shows that the minimum time wastage in PMTUD is the one degree polynomial equation.

\begin{theorem} \label{theorem_3}
The minimum total time wastage  for nodes dropping packet at least in one position in the path  between source and destination in PMTUD is given by,

\begin{equation}
     {^p}T_{w}  = \sum_{i=1}^{\ a }n_i({^p}T_d\ +\ {^p}T'_d)+aT_f
\end{equation}

Where, \(T_f\) is the time of fragmentation by source node and a is number of times packet is dropped in the path.
\end{theorem}

\begin{proof}
In Theorem \ref{theorem_2}, we have \({^p}T_{w}\ =n_1 ({^p}T_d+{^p}T'_d) +T_f\) which is time wastage for a node dropping packet in the path between the source and destination. The time wastage for two nodes dropping packets in the path is given by \((n_1+n_2)({^p}T_d+{^p}T'_d) +2T_f\).

Where,$\ n_1 \in\ n\ $ and $\  n_2\ \in n - (1, 2, 3, ... , n_1)\ , \ n_1 \ne\ n_2\ , \  n_1<n_2 ,\ $ such that $n \in(nodes) \to {Z}^+$ are first and second nodes respectively, that drops packet in the path. Similarly, as we go on increasing the number of nodes that are dropping the packet to $i^th - term$, their is always a term \(n_i({^p}T_d+{^p}T'_d)+T_f\) is incremented to the preceding time wastage. So in general for all of nodes which drops packet in the path in varying order can be given by:

\begin{align}\label{3_equ}
     {^p}T_{w}  =   \left (\begin{array}{lll}({^p}T_d\ +\ {^p}T'_d)\sum_{i=1}^{\ a } [n_i]\end{array} \right )+\ aT_f , \\ \quad
  \left (\begin{array}{lll}
  {n_{i-1} \ne\ n_i \ne n_{i+1},} \\ {  n_i \in (a)\  \to (n),}\\ {  n_{i+1} \in (a)-(1,2,3....n_i)}\end{array}\right.
\end{align}

Where $'a'$ is number of times a packet is dropped in a transmission. The Equation \ref{3_equ} is a general equation of total time wastage for the packet that is dropped by any arbitrary node upto $n^{th}-node$, \(where\ n\ \in (nodes) \to {Z}^+.\)
\end{proof}

\subsection{Analysis between Worst Case and best case Scenario}
\begin{theorem}
The minimum time wastage equals to the maximum total time wastage at a = n using PMTUD algorithm, which is given by, 
\begin{equation}
    {^p}T_{w}\ =S_n({^p}T_d+{^p}T'_d) + nT_f.
\end{equation}
\end{theorem}
\begin{proof}
If all the nodes in between the path of source and destination drops the packet consecutively then from Theorem \ref{theorem_3}:
\begin{gather} \label{4.1_equ}
     {^p}T_{w} = ({^p}T_d\ +\ {^p}T'_d)\sum_{i=1}^{\ n } [n_i]+nT_f\\
     where\ a\ =\ n\ ,\  as\ all\ nodes\ are\ dropping\ packet\notag
\end{gather}

Expanding Equation \ref{4.1_equ} by putting values of $'i'$ which runs from 1 to n, as all of the nodes upto $n$ drops the packet consecutively and is given by:
\begin{align}\label{4_eqn}
{^p}T_{w}=&{n_1 ({^p}T_d+{^p}T'_d) +T_f}+{n_2({^p}T_d+{^p}T'_d)+T_f} \notag\\
&+ \ldots +{n_n ({^p}T_d+{^p}T'_d)+T_f}\notag\\
=&({^p}T_d+{^p}T'_d)(n_1 + n_2 +...+ n_n) + nT_f\notag\\
=&({^p}T_d+{^p}T'_d) (1 + 2 + 3 + 4 +...+ n) + nT_f\notag\\
=&({^p}T_d+{^p}T'_d)S_n + nT_f\notag\\
{^p}T_{w}\ = &S_n({^p}T_d+{^p}T'_d) + nT_f
\end{align}

Since the Equation \ref{4_eqn} which is derived from General formula of minimum total time wastage is equal to the maximum total time wastage at a = n.

\end{proof}
\begin{theorem}
The limit of lower bound and the Upper bound of the total time wastage using PMTUD algorithm is given by :
\begin{align}
 n_1({^p}T_d +{^p}T'_d) + T_f \le {^p}T_{w} \le S_n({^p}T_d+{^p}T'_d) + nT_f
\end{align}
\end{theorem}
\begin{proof}
In Theorem \ref{theorem_3} the general formula for the nodes dropping packets at least at one position is given by :
\begin{align}
{^p}T_{w}=({^p}T_d\ +\ {^p}T'_d)\sum_{i=1}^{\ a } [n_i]+nT_f\\
\textit{At a= 1 ,}\qquad
{^p}T_{w} = n_1({^p}T_d\ +\ {^p}T'_d)+T_f \label{e}
\end{align}

Which is the lowest limit of the time wastage\(({^p}T_{w})\), when a single node drops a packet by any arbitrary node \(n_1^{th} \) in the path. As we go on increasing the value of $'a'$ (frequency of nodes participating in packet drop) until \(a=n\), then \({^p}T_{w}\) reaches up-to a certain point, given as.: 
\begin{align}\label{r}
{^p}T_{w} = S_n({^p}T_d+{^p}T'_d) + nT_f
\end{align}
Which is the upper limit of the time wastage. From Equation~\ref{e} \& \ref{r} we can define the limits of the time wastage as:
\begin{equation}\label{b}
n_1({^p}T_d +{^p}T'_d) + T_f\ \le\  {^p}T_{w} \ \le\ S_n({^p}T_d+{^p}T'_d) + nT_f
\end{equation}

\textit{Therefore, from Equation~\ref{b} the value of \({^p}T_{w}\) can't be lower then \(n_1({^p}T_d+{^p}T'_d) +T_f\) and can't be higher then \(S_n({^p}T_d+{^p}T'_d) + nT_f.\)}
\end{proof}

Now their must be a question why don't we implement the models in simulator? This study is completely based on theoretical assumption of the real world network analysis, and the main motivations and focus of the paper is to initiate a new method, by introducing the theoretical mathematics to measure the effect of PMTUD in new generation networks i.e IPv4 \& IPv6 network, mainly with the time delay and to provide as a basics for further growth in research studies in such area. As many studies that has been published by journals in this area, haven't presented or discussed any mathematical aspects of the PMTUD in IPv4 \& IPv6 network rather then they have focused more in the analysis, surveys and a new proposals of techniques which of-course is important but we should make a way for such new mathematical modelling of PMTUD parameters.

\section{Analysis \& Results}\label{sec:analysis_result}
\subsection{Levels in Best Case Scenario}
There is a huge variation in time delay about the location of the node where the packet is dropped. Our further analysis on the Theorem \ref{theorem_2} \& \ref{theorem_3} results in different output results.

The time wastage due to a packet drop has significantly varies with the position of the node which drops the packet. i.e. the packet dropped near the source result in minimal time delay where as the packet drop near the destination resulted in maximal time delay while the packet drop at the centre of the source and destination result in an average of the two time delays. In other words the time wastage when the packet is dropped while moving from source node towards the destination node increases linearly.
 
This can be verified by taking an experimental study between the packet dropped by different nodes and their time wastage. 
\subsection{Only in one position}
For this we have taken the Theorem \ref{theorem_2} i.e 
\[{^p}T_w = n_1 ({^p}T_d + {^p}T'_d ) + T_f , where\ n \subset Z^+\]
We have taken the above theorem and apply it starting from source node and moving towards destination node consecutively and we found a variation in time wastage. For the computing purpose in a simulation we have used predefined values for the parameters. i.e.
${^p}T_d = 1ms$, ${^p}T'_d = 0.5ms$, $T_f = 1ms$, $n=17$.

The Table \ref{tab} shows the variation in time wastage due to a single packet drop in different positions between the source and the destination. The same data is illustrated in bar graph in Figure \ref{net_topo1}. In the Figure \ref{net_topo1}. as the node position increases the time wastage increases linearly. Further the time wastage near the source node is least and at destination its maximum. From this data and experiment we came to this conclusion that the rate of packet drop if happens nearer to the source have very less impact on the time wastage as that of the packet dropping near to the destination node. 
In practical, this analysis helps in designing the network topology in Server to client communication in such a way that the nodes that are prone to or inclined to packet drops should be kept nearer to Server side which can help in decreased time wastage compared to the irregular distribution of the problem node or keeping them near the client side. 

\begin{table}[!t]
\caption{The time wastage using pre-defined values for parameter ${^p}T_d = 1ms$, ${^p}T'_d = 0.5ms$, $T_f = 1ms$, $n=17$.  \label{tab}}
\resizebox{\columnwidth}{!}{
\begin{tabular}{m{5em}m{8em}m{7em}}
\hline
{\centering Position of node $(n_1)$} & {\centering Time Wastage $(T_w)$} & {\centering Value of $T_w$ (ms)}\\ [3ex]
\hline
1 & ${^p}T_{d} + {^p}T'_{d} + T_{f}$ & 2.5 ms\\

2 & $2({^p}T_{d} + {^p}T'_{d}) + T_{f}$ & 4.0 ms\\ 

3 & $3({^p}T_{d} + {^p}T'_{d}) + T_{f}$ & 5.5 ms\\

4 & $4({^p}T_{d} + {^p}T'_{d}) + T_{f}$ & 7.0 ms\\

5 & $5({^p}T_{d} + {^p}T'_{d}) + T_{f}$ & 8.5 ms\\

6 & $6({^p}T_{d} + {^p}T'_{d}) + T_{f}$ & 10.0 ms\\

7 & $7({^p}T_{d} + {^p}T'_{d}) + T_{f}$ & 11.5 ms\\

8 & $8({^p}T_{d} + {^p}T'_{d}) + T_{f}$ & 13.0 ms\\

9 & $9({^p}T_{d} + {^p}T'_{d}) + T_{f}$ & 14.5 ms\\

10 & $10({^p}T_{d} + {^p}T'_{d}) + T_{f}$ & 16.0 ms\\

11 & $11({^p}T_{d} + {^p}T'_{d}) + T_{f}$ & 17.5 ms\\

12 & $12({^p}T_{d} + {^p}T'_{d}) + T_{f}$ & 19.0 ms\\

13 & $13({^p}T_{d} + {^p}T'_{d}) + T_{f}$ & 20.5 ms\\

14 & $14({^p}T_{d} + {^p}T'_{d}) + T_{f}$ & 22.0 ms\\

15 & $15({^p}T_{d} + {^p}T'_{d}) + T_{f}$ & 23.5 ms\\

16 & $16({^p}T_{d} + {^p}T'_{d}) + T_{f}$ & 25.0 ms\\

17 & $17({^p}T_{d} + {^p}T'_{d}) + T_{f}$  & 26.5 ms\\[1ex] 
\hline
\end{tabular}}
\end{table}

This can be described in the bar graph in Figure \ref{net_topo1} of packets dropped at different positions of node and their respective time wastage. The graph shows an increase in the time wastage as we move from source node towards the destination node.

 \begin{figure}[!t]
 \centering
 \centerline{\includegraphics[width=1\columnwidth, height=0.8\columnwidth]{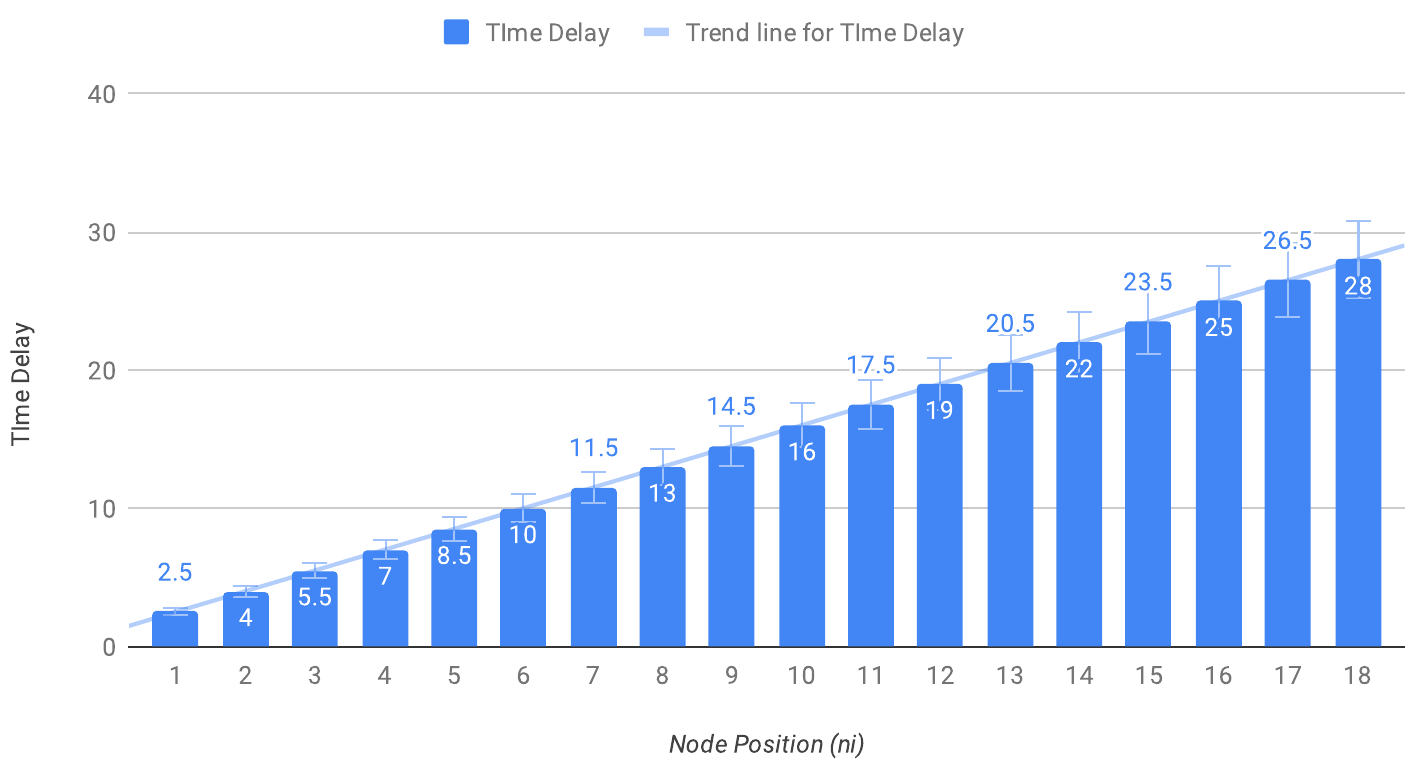}}
 \caption{ Increasing time wastage with the increase in the position of the node where a single packet is dropped i.e $n_i$.}
 \label{net_topo1}
 \end{figure}
 
\subsection{At-least in one position}
In Theorem \ref{theorem_3} we have the time wastage for packet drop frequency at least in one-position as:
\begin{align}
^{p}T_{w}(n,a) = \sum_{i=1}^{\ a }n_i({^p}T_{d}+{^p}T'_{d}) + aT_f
\end{align}
The above equation on using the subset $b_{n_k}$ can be written as:
\begin{align}
^{p}T_{w}(n,a) = \sum_{k=1}^{\ a }b_{n_k}({^p}T_{d}+{^p}T'_{d}) + aT_f
\end{align}
The subset $(b_{n_k})_{k=1}^{a}$ gives the number of possible combinations of nodes contributing for dropping a packet $a-times$ in total. The location of nodes can be represented as a arithmetic progression with $first-term$ as $'1'$ \& $distance=1$ as: 
\[B=\{1,2,3, \cdots\};\]
The arithmetic progression can be expressed by the following notation of sequences as:
\[(b_n)_{n=1}^{l} = \{b_1,b_2,b_3, \cdots, b_l\}\]
Where $l$ represents the last term in the arithmetic progression. 

Lets take an example, where a packet is transmitted through a path with 5 nodes from source to destination. now lets assume that the packets is dropped twice until complete transmission. The first drop happens at node $2$ and then the second transmission the node can't be dropped by node below 2 but can only be above $2$ i.e \{3,4,5\}. Let the second drop will be by 4 and after this the packet is transmitted to the destination. We see the packet is dropped twice during the transmission by nodes \{2,4\}, which means the frequency of the packet drop by nodes is $a=2$. Using the Theorem \ref{theorem_3} we can find the time wastage i.e. $^pT_w$ by PMTUD in transmitting the packet. Now from the given value of $a=2$ i.e frequency of packet drop by node the set $A = \{1,2,3,4,5\}$ can have 10 different combinations or ways in which the packet can be dropped by intermediate node for a given value $a=2$ i.e:\\

(1,2), (1,3), (1,4), (1,5), (2,3), (2,4), (2,5), (3,4), (3,5), (4,5);\\

The above relation is similar to the a-combination of set A without repetition which gives $\binom{5}{2}$ number of combinations i.e:
\[\binom{5}{2} = {5!\over {(5-2)!2!}} = 10\ combinations\]
Similarly, if the frequency of packet drop is $a=3$ and $a=4$ which forms $6,2$ combinations of set A which can also be given by a-combinations as:
\[\binom{5}{3} = {5!\over {(5-3)!3!}} = 6\ combinations\]
\[\binom{5}{4} = {5!\over {(5-4)!4!}} = 2\ combinations\]

In general, for $n$ number of nodes between source and destination for a given frequency of packet drop $a\le n,\ where\ a\in \mathrm{N}$, can be defined in-terms of a-combinations formed from the sequence $B$ of size $n$ without repetition. While a-combination means the combination of a-objects. While the objects in the combinations signifies the position of node in set B. The a-combination of set B made upto $\binom{n}{a}$ number of sub-sequences of sequence B. In other terms, the number of ways the packet is drooped by the nodes is $\binom{n}{a}$.
In General, the following relation is used to defined the possible combinations in form of sub-sequences i.e:
\[B'=(b_{n_i})_{i=1}^{a}=\{b_{n_1},b_{n_1},b_{n_1}, \cdots, b_{n_a}\}\]
\[where\ b_{n_i} \subset B\ \&\ n(b_{n_i}) = a\]
While it should be noted that $n_{k}<n_{k+1}$ which means $n_k$ monotonically increasing. Further $n_k$ is random in nature \& the domain is $n_k\in[1,n]$.
The Set $B'$ is a sequence of sub-sequences of the sequence $B$ with size $a$. The term $b_{n_i}$ is a object of the sub-sequence B from the combination of $a$ objects in the set $B$ with $n$ objects without repetition.
Therefore, the time wastage for all of these combinations can be defined from theorem 2 as:
\begin{align}
^{p}T_{w}(n,a) = \sum_{k=1}^{\ a }b_{n_k}({^p}T_{d}+{^p}T'_{d}) + aT_f
\end{align}
Now it should be noted that there are $\binom{n}{a}$ combinations which means there are same number of ways the packet is dropped by $n$ nodes in frequency of $a$


The term ${^p}T_w(n,a)$ is the corresponding time wastage of sub-sequence $(b_{n_i})_{i=1}^{a}$ of a-combinations of the sequence of $n$ nodes. Through this equation we can find the time wastage of nodes dropping a packet in $\binom{n}{a}$ different ways. Now the use of such an method gives us different outputs on time wastage. The summation of objects of subsequences have lower and upper limit which is given in Theorem \ref{theorem_6}.
\begin{theorem}\label{theorem_6}
If a sequence $b_{n_1},b_{n_2}, \cdots, b_{n_a}$ of size 'a' is a sub-sequence of  sequence $b_1, b_2, \cdots, b_n$ which is strictly increasing then
\[\sum_{n=1}^{a}b_n \le \sum_{i=1}^{a}b_{n_i} \le \frac{a}{2}(2n-a+1) \]
where $b_n \in \mathbb{Z^{+}}$ \& $n_i, b_{n_i}$ are strictly increasing.
\end{theorem}
In the Theorem \ref{theorem_6} its is clear that the summation of objects of subsequence of B increases gradually but not linearly. This is because there are subsequences of B where the sum of the objects is equal alternative way. Since this will have similar effects on the respective time wastage and have the following upper and lower limits according to the sequence followed from the dropping of packet by intermediate node as:
\begin{align}
\max_{a \to n}{^p}T_w(n,a) &= \sum_{i=1}^{a}(n-i+1)({^p}T_{d} + {^p}T'_{d}) + aT_{f}\\
&= \frac{a}{2}(2n-a+1)({^p}T_{d} + {^p}T'_{d}) + aT_{f}< \epsilon \label{max}
\end{align}
$iff\ \lim_{i\to n}b_{n_i} = b_n\ then,$
\begin{align}
\min_{a \to 0}{^p}T_w(n,a) &= \sum_{i=1}^{a}(b_n)({^p}T_{d} + {^p}T'_{d}) + aT_{f}\\
&= \frac{a}{2}(a+1)({^p}T_{d} + {^p}T'_{d}) + aT_{f}< \epsilon \label{min}\\ \notag
where\ \epsilon > 0
\end{align}
\begin{align} \label{avgT}
avg{^p}T_w(n_a) = \frac{a}{2}(n+1)({^p}T_{d} + {^p}T'_{d}) + aT_{f}
\end{align}

The Summation $\sum_{i=1}^{a}b_{n_i}$ for a given frequency $a$ has same value for some $i\in a$. This can be represented by taking an example of a sequence with $n=6$ and frequency $a=2$.
\begin{theorem}
For a given value of n their will be $a(n-a)+1$ number of distinct summations for $\binom{n}{a}$ subsequence $(b_{n_i})_{i=1}^{a}$ of the sequence $b_n$ for a given value of frequency $a$ which is represented by $(S_{i})_{i=0}^{a(n-a)+1}$ where $S_{i}$ is monotonically increasing.
\end{theorem}
\begin{proof}
Let a sequence $B = \{b_1,b_2,b_3, \cdots, b_n\}$ with $n$ objects. Then the $\binom{n}{a}$ number of combination without repetition of sequence $B$ is given by ${(b_{n_i})}_{i=1}^{a} = \{b_{n_1}, b_{n_2}, b_{n_3}, \cdots, b_{n_a}\}$. Then the distinct sum of the objects of sub-sequences $b_{n_i}$ is given by notation $S_{i}$. Then the initial sum is the $\min {(b_{n_i})}_{i=1}^{a} = {1+2+3+4+ \cdots + a} = \sum_{i=1}^{a}i$ and the last sum is $\max {(b_{n_i})}_{i=1}^{a} = {n-a+1, n-a+2, n-a+3 \cdots n} = \sum_{i=1}^{a}(n-i+1)$. Now the numbers of distinct sums is given by:
\begin{align}
nD(S_{i}) &= \sum_{i=1}^{a}(n-i+1) - \sum_{i=1}^{a}(i +1)\\
nD(S_{i}) &= a(n-a)+1
\end{align}

Therefore, from the above value the range of $S_i$ is i = {$1$ to $a(n-a)+1$} and the the sequence is represented as $(S_i)_{i=1}^{a(n-a)+1}$. 
\end{proof}

Lets, take an arithmetic progression $AP= \{1,2,3,4,5,6\}$ with $n=6$ \& the frequency of packet drop be $a=2$ and representing it by sequence B = $(b_n)_{i=1}^{n} = \{b_1,b_2,b_3,b_4,b_5,b_6\}$, then their will be $\binom{6}{2}$ number of sub-sequences of sequence AP i.e. B' = $\{(b_{n_k})_{i=1}^{2}\}=$ \{(1,2), (1,3), (1,4), (1,5), (1,6), (2,3), (2,4), (2,5), (2,6), (3,4), (3,5), (3,6), (4,5), (4,6), (5,6)\}. The distinct summation values of the objects of the sub-sequences of AP range from: 
\[(S_i)_{i=1}^{a(n-a)+1} = (S_i)_{i=1}^{9}\]
i.e\hfill\[(S_i)_{i=1}^{9} = \{S_1,S_2,S_3,S_4,S_5,S_6,S_7,S_8,S_9\}\]

 \begin{table}[!t]
\caption{Distinct Summations of sub-sequences at $a=2$ \& $n=6$\label{tab6}}
\resizebox{\columnwidth}{!}{%
\begin{tabular}{c|ccccccccc}
\textbf{Range (i)} & \textbf{$1$} & \textbf{$2$} & \textbf{$3$} & \textbf{$4$} & \textbf{$5$} & \textbf{$6$} & \textbf{$7$} & \textbf{$8$} & \textbf{$9$} \\
\hline
\textbf{\(S_i\)} & $3$ & $4$ & $5$ & $6$ & $7$ & $8$ & $9$ & $10$ & $11$\\
\hline
\textbf{Frequency} & $1$ & $1$ & $2$ & $2$ & $3$ & $2$ & $2$ & $1$ & $1$
\end{tabular}
}
\end{table}

 \begin{table}[!t]
\caption{Distinct Summations of sub-sequences at $a=3$ \& $n=6$\label{tab4}}
\resizebox{\columnwidth}{!}{%
\begin{tabular}{c|cccccccccc}
\textbf{Range (i)} & \textbf{$1$} & \textbf{$2$} & \textbf{$3$} & \textbf{$4$} & \textbf{$5$} & \textbf{$6$} & \textbf{$7$} & \textbf{$8$} & \textbf{$9$} & \textbf{$10$} \\
\hline
\textbf{\(S_i\)} & $6$ & $7$ & $8$ & $9$ & $10$ & $11$ & $12$ & $13$ & $14$ & $15$\\
\hline
\textbf{Frequency} & $1$ & $1$ & $3$ & $3$ & $4$ & $4$ & $3$ & $3$ & $1$  & $1$
\end{tabular}
}
\end{table}

\begin{figure}
    \centering
    \includegraphics[width=1\columnwidth, height=0.8\columnwidth]{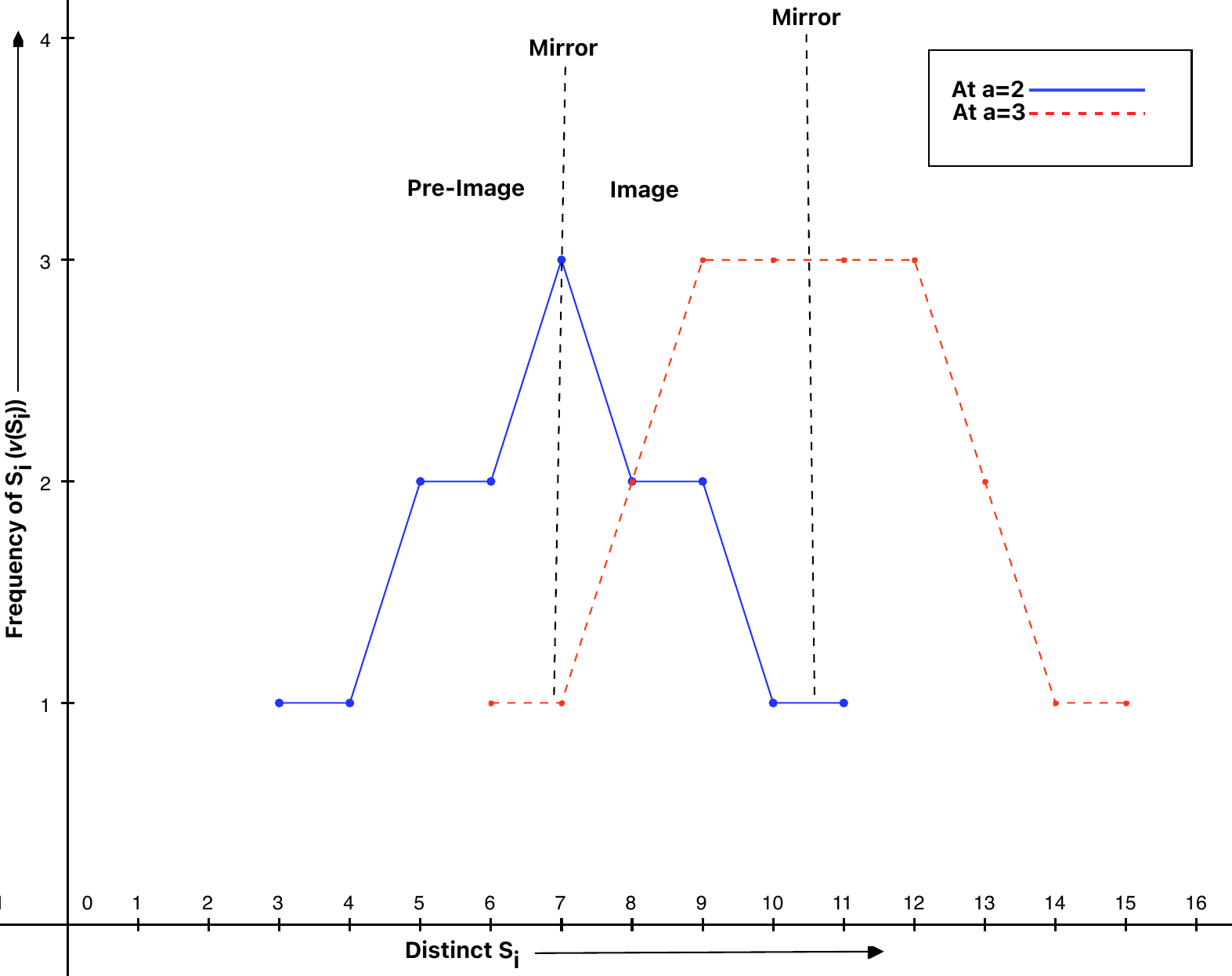}
    \caption{Relation between the distinct sum and there frequencies at $a=2$ \& $a=3$ for $n=6$}
    \label{fig:a23}
\end{figure}

\begin{figure}[!t]
    \centering
    \includegraphics[width=1\columnwidth, height=0.8\columnwidth]{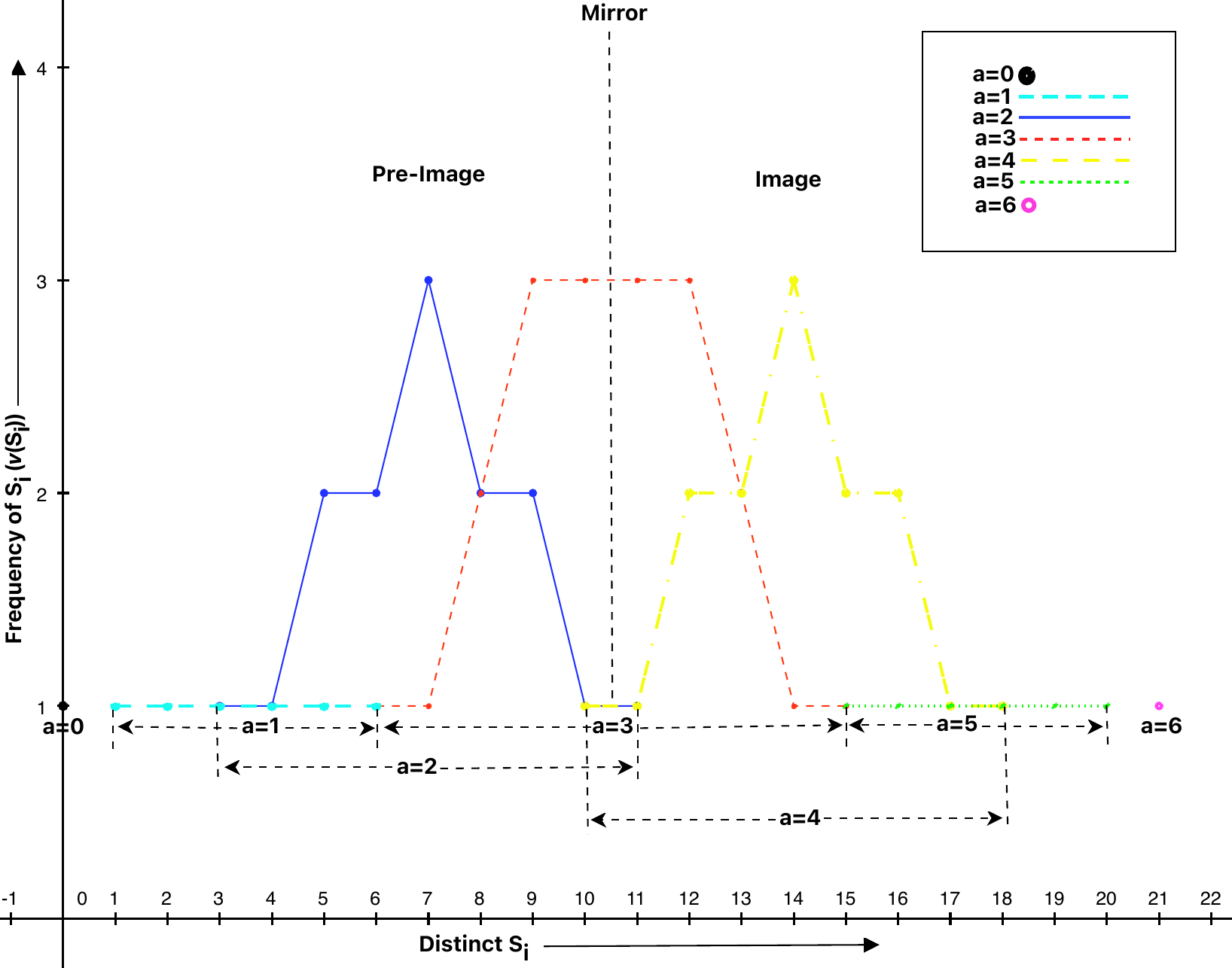}
    \caption{Relation between the distinct sum and there frequencies from $a=0$ to $a=6$ for $n=6$}
    \label{fig:a06}
\end{figure}

 \begin{table}[!t]
\caption{Distinct Summations of sub-sequences at $a=0$ to $a=6$ at $n=6$\label{taba06}}
\resizebox{\columnwidth}{!}{%
\begin{tabular}{c|c|ccccccccccc}
$(a)_{i=1}^{6}$ & \textbf{i} & \textbf{$1$} & \textbf{$2$} & \textbf{$3$} & \textbf{$4$} & \textbf{$5$} & \textbf{$6$} & \textbf{$7$} & \textbf{$8$} & \textbf{$9$} & \textbf{$10$} \\\cline{1-12}

\multirow{2}{*}{$a=0$} & \textbf{\(S_i\)} & $0$ & $-$ & $-$ & $-$ & $-$ & $-$ & $-$ & $-$ & $-$ & $-$\\\cline{3-12}
& $\nu(S_i)$ & $1$ & $-$ & $-$ & $-$ & $-$ & $-$ & $-$ & $-$ & $-$  & $-$\\
\hline
\multirow{2}{*}{$a=1$} & \textbf{\(S_i\)} & $1$ & $2$ & $3$ & $4$ & $5$ & $6$ & $-$ & $-$ & $-$ & $-$ \\\cline{3-12}
& $\nu(S_i)$ & $1$ & $1$ & $1$ & $1$ & $1$ & $1$ & $-$ & $-$ & $-$  & $-$\\

\hline
\multirow{2}{*}{$a=2$} & \textbf{\(S_i\)} & $3$ & $4$ & $5$ & $6$ & $7$ & $8$ & $9$ & $10$ & $11$ & $-$\\\cline{3-12}
& $\nu(S_i)$ & $1$ & $1$ & $2$ & $2$ & $3$ & $2$ & $2$ & $1$ & $1$  & $-$\\

\hline
\multirow{2}{*}{$a=3$} & \textbf{\(S_i\)} & $6$ & $7$ & $8$ & $9$ & $10$ & $11$ & $12$ & $13$ & $14$ & $15$\\\cline{3-12}
& $\nu(S_i)$ & $1$ & $1$ & $2$ & $3$ & $3$ & $3$ & $3$ & $2$ & $1$  & $1$\\

\hline
\multirow{2}{*}{$a=4$} & \textbf{\(S_i\)} & $10$ & $11$ & $12$ & $13$ & $14$ & $15$ & $16$ & $17$ & $18$ & $-$\\\cline{3-12}
& $\nu(S_i)$ & $1$ & $1$ & $2$ & $2$ & $3$ & $2$ & $2$ & $1$ & $1$  & $-$\\

\hline
\multirow{2}{*}{$a=5$} & \textbf{\(S_i\)} & $15$ & $16$ & $17$ & $18$ & $19$ & $20$ & $-$ & $-$ & $-$ & $-$\\\cline{3-12}
& $\nu(S_i)$ & $1$ & $1$ & $1$ & $1$ & $1$ & $1$ & $-$ & $-$ & $-$  & $-$\\

\hline
\multirow{2}{*}{$a=6$} & \textbf{\(S_i\)} & $21$ & $-$ & $-$ & $-$ & $-$ & $-$ & $-$ & $-$ & $-$ & $-$\\\cline{3-12}
& $\nu(S_i)$ & $1$ & $-$ & $-$ & $-$ & $-$ & $-$ & $-$ & $-$ & $-$  & $-$\\

\end{tabular}
}
\end{table}

 \begin{table}[!t]
\caption{Distinct Summations of sub-sequences at $(n)_{1}^{6}\ \&\ (a)_{0}^{n}$, where $t$ $\in\ \mathbb{N}\ \&\ t=a(n-a)+1 $\label{taba0n}}
\resizebox{\columnwidth}{!}{%
\begin{tabular}{c|c|c|cccccccccc}
{$(n)_{0}^{6}$} & $(a)_{i=1}^{n}$ & \textbf{$(i)_{1}^{t}$} & \textbf{$1$} & \textbf{$2$} & \textbf{$3$} & \textbf{$4$} & \textbf{$5$} & \textbf{$6$} & \textbf{$7$} & \textbf{$8$} & \textbf{$9$} & \textbf{$10$} \\\hline\hline
\multirow{4}{*}{$n=1$}&\multirow{2}{*}{$a=0$} & \textbf{\(S_i\)} & $0$ & $-$ & $-$ & $-$ & $-$ & $-$ & $-$ & $-$ & $-$ & $-$\\
&& $\nu(S_i)$ & $1$ & $-$ & $-$ & $-$ & $-$ & $-$ & $-$ & $-$ & $-$  & $-$\\
\cline{2-13}
&\multirow{2}{*}{$a=1$} & \textbf{\(S_i\)} & $1$ & $-$ & $-$ & $-$ & $-$ & $-$ & $-$ & $-$ & $-$ & $-$ \\
&& $\nu(S_i)$ & $1$ & $-$ & $-$ & $-$ & $-$ & $-$ & $-$ & $-$ & $-$  & $-$\\

\hline\hline
\multirow{6}{*}{$n=2$}&\multirow{2}{*}{$a=0$} & \textbf{\(S_i\)} & $0$ & $-$ & $-$ & $-$ & $-$ & $-$ & $-$ & $-$ & $-$ & $-$\\
&& $\nu(S_i)$ & $1$ & $-$ & $-$ & $-$ & $-$ & $-$ & $-$ & $-$ & $-$  & $-$\\
\cline{2-13}
&\multirow{2}{*}{$a=1$} & \textbf{\(S_i\)} & $1$ & $2$ & $-$ & $-$ & $-$ & $-$ & $-$ & $-$ & $-$ & $-$ \\
&& $\nu(S_i)$ & $1$ & $1$ & $-$ & $-$ & $-$ & $-$ & $-$ & $-$ & $-$  & $-$\\

\cline{2-13}
&\multirow{2}{*}{$a=2$} & \textbf{\(S_i\)} & $3$ & $-$ & $-$ & $-$ & $-$ & $-$ & $-$ & $-$ & $-$ & $-$\\
&& $\nu(S_i)$ & $1$ & $-$ & $-$ & $-$ & $-$ & $-$ & $-$ & $-$ & $-$  & $-$\\
\hline\hline
\multirow{8}{*}{$n=3$}&\multirow{2}{*}{$a=0$} & \textbf{\(S_i\)} & $0 $ & $-$ & $-$ & $-$ & $-$ & $-$ & $-$ & $-$ & $-$ & $-$\\
&& $\nu(S_i)$ & $1$ & $-$ & $-$ & $-$ & $-$ & $-$ & $-$ & $-$ & $-$  & $-$\\
\cline{2-13}
&\multirow{2}{*}{$a=1$} & \textbf{\(S_i\)} & $1$ & $2$ & $3$ & $-$ & $-$ & $-$ & $-$ & $-$ & $-$ & $-$ \\
&& $\nu(S_i)$ & $1$ & $1$ & $1$ & $-$ & $-$ & $-$ & $-$ & $-$ & $-$  & $-$\\

\cline{2-13}
&\multirow{2}{*}{$a=2$} & \textbf{\(S_i\)} & $3$ & $4$ & $5$ & $-$ & $-$ & $-$ & $-$ & $-$ & $-$ & $-$\\
&& $\nu(S_i)$ & $1$ & $1$ & $1$ & $-$ & $-$ & $-$ & $-$ & $-$ & $-$  & $-$\\

\cline{2-13}
&\multirow{2}{*}{$a=3$} & \textbf{\(S_i\)} & $6$ & $-$ & $-$ & $-$ & $-$ & $-$ & $-$ & $-$ & $-$ & $-$\\
&& $\nu(S_i)$ & $1$ & $-$ & $-$ & $-$ & $-$ & $-$ & $-$ & $-$ & $-$  & $-$\\
\hline\hline
\multirow{10}{*}{$n=4$}&\multirow{2}{*}{$a=0$} & \textbf{\(S_i\)} & $0$ & $-$ & $-$ & $-$ & $-$ & $-$ & $-$ & $-$ & $-$ & $-$\\
&& $\nu(S_i)$ & $1$ & $-$ & $-$ & $-$ & $-$ & $-$ & $-$ & $-$ & $-$  & $-$\\
\cline{2-13}
&\multirow{2}{*}{$a=1$} & \textbf{\(S_i\)} & $1$ & $2$ & $3$ & $4$ & $-$ & $-$ & $-$ & $-$ & $-$ & $-$ \\
&& $\nu(S_i)$ & $1$ & $1$ & $1$ & $1$ & $-$ & $-$ & $-$ & $-$ & $-$  & $-$\\

\cline{2-13}
&\multirow{2}{*}{$a=2$} & \textbf{\(S_i\)} & $3$ & $4$ & $5$ & $6$ & $7$ & $-$ & $-$ & $-$ & $-$ & $-$\\
&& $\nu(S_i)$ & $1$ & $1$ & $2$ & $1$ & $1$ & $-$ & $-$ & $-$ & $-$  & $-$\\

\cline{2-13}
&\multirow{2}{*}{$a=3$} & \textbf{\(S_i\)} & $6$ & $7$ & $8$ & $9$ & $-$ & $-$ & $-$ & $-$ & $-$ & $-$\\
&& $\nu(S_i)$ & $1$ & $1$ & $1$ & $1$ & $-$ & $-$ & $-$ & $-$ & $-$  & $-$\\

\cline{2-13}
&\multirow{2}{*}{$a=4$} & \textbf{\(S_i\)} & $10$ & $-$ & $-$ & $-$ & $-$ & $-$ & $-$ & $-$ & $-$ & $-$\\
&& $\nu(S_i)$ & $1$ & $-$ & $-$ & $-$ & $-$ & $-$ & $-$ & $-$ & $-$  & $-$\\
\hline\hline
\multirow{12}{*}{$n=5$}&\multirow{2}{*}{$a=0$} & \textbf{\(S_i\)} & $0$ & $-$ & $-$ & $-$ & $-$ & $-$ & $-$ & $-$ & $-$ & $-$\\
&& $\nu(S_i)$ & $1$ & $-$ & $-$ & $-$ & $-$ & $-$ & $-$ & $-$ & $-$  & $-$\\
\cline{2-13}
&\multirow{2}{*}{$a=1$} & \textbf{\(S_i\)} & $1$ & $2$ & $3$ & $4$ & $5$ & $-$ & $-$ & $-$ & $-$ & $-$ \\
&& $\nu(S_i)$ & $1$ & $1$ & $1$ & $1$ & $1$ & $-$ & $-$ & $-$ & $-$  & $-$\\

\cline{2-13}
&\multirow{2}{*}{$a=2$} & \textbf{\(S_i\)} & $3$ & $4$ & $5$ & $6$ & $7$ & $8$ & $9$ & $-$ & $-$ & $-$\\
&& $\nu(S_i)$ & $1$ & $1$ & $2$ & $2$ & $2$ & $1$ & $1$ & $-$ & $-$  & $-$\\

\cline{2-13}
&\multirow{2}{*}{$a=3$} & \textbf{\(S_i\)} & $6$ & $7$ & $8$ & $9$ & $10$ & $11$ & $12$ & $-$ & $-$ & $-$\\
&& $\nu(S_i)$ & $1$ & $1$ & $2$ & $2$ & $2$ & $1$ & $1$ & $-$ & $-$  & $-$\\

\cline{2-13}
&\multirow{2}{*}{$a=4$} & \textbf{\(S_i\)} & $10$ & $11$ & $12$ & $13$ & $14$ & $-$ & $-$ & $-$ & $-$ & $-$\\
&& $\nu(S_i)$ & $1$ & $1$ & $1$ & $1$ & $1$ & $-$ & $-$ & $-$ & $-$  & $-$\\

\cline{2-13}
&\multirow{2}{*}{$a=5$} & \textbf{\(S_i\)} & $15$ & $-$ & $-$ & $-$ & $-$ & $-$ & $-$ & $-$ & $-$ & $-$\\
&& $\nu(S_i)$ & $1$ & $-$ & $-$ & $-$ & $-$ & $-$ & $-$ & $-$ & $-$  & $-$\\
\hline\hline
\multirow{14}{*}{$n=6$}&\multirow{2}{*}{$a=0$} & \textbf{\(S_i\)} & $0$ & $-$ & $-$ & $-$ & $-$ & $-$ & $-$ & $-$ & $-$ & $-$\\
&& $\nu(S_i)$ & $1$ & $-$ & $-$ & $-$ & $-$ & $-$ & $-$ & $-$ & $-$  & $-$\\
\cline{2-13}
&\multirow{2}{*}{$a=1$} & \textbf{\(S_i\)} & $1$ & $2$ & $3$ & $4$ & $5$ & $6$ & $-$ & $-$ & $-$ & $-$ \\
&& $\nu(S_i)$ & $1$ & $1$ & $1$ & $1$ & $1$ & $1$ & $-$ & $-$ & $-$  & $-$\\

\cline{2-13}
&\multirow{2}{*}{$a=2$} & \textbf{\(S_i\)} & $3$ & $4$ & $5$ & $6$ & $7$ & $8$ & $9$ & $10$ & $11$ & $-$\\
&& $\nu(S_i)$ & $1$ & $1$ & $2$ & $2$ & $3$ & $2$ & $2$ & $1$ & $1$  & $-$\\

\cline{2-13}
&\multirow{2}{*}{$a=3$} & \textbf{\(S_i\)} & $6$ & $7$ & $8$ & $9$ & $10$ & $11$ & $12$ & $13$ & $14$ & $15$\\
&& $\nu(S_i)$ & $1$ & $1$ & $2$ & $3$ & $3$ & $3$ & $3$ & $2$ & $1$  & $1$\\

\cline{2-13}
&\multirow{2}{*}{$a=4$} & \textbf{\(S_i\)} & $10$ & $11$ & $12$ & $13$ & $14$ & $15$ & $16$ & $17$ & $18$ & $-$\\
&& $\nu(S_i)$ & $1$ & $1$ & $2$ & $2$ & $3$ & $2$ & $2$ & $1$ & $1$  & $-$\\

\cline{2-13}
&\multirow{2}{*}{$a=5$} & \textbf{\(S_i\)} & $15$ & $16$ & $17$ & $18$ & $19$ & $20$ & $-$ & $-$ & $-$ & $-$\\
&& $\nu(S_i)$ & $1$ & $1$ & $1$ & $1$ & $1$ & $1$ & $-$ & $-$ & $-$  & $-$\\

\cline{2-13}
&\multirow{2}{*}{$a=6$} & \textbf{\(S_i\)} & $21$ & $-$ & $-$ & $-$ & $-$ & $-$ & $-$ & $-$ & $-$ & $-$\\
&& $\nu(S_i)$ & $1$ & $-$ & $-$ & $-$ & $-$ & $-$ & $-$ & $-$ & $-$  & $-$\\
\hline\hline
\end{tabular}
}
\end{table}

  \begin{table}[!t]
\caption{General method for computing Distinct Summations of sub-sequences at $(a)_{0}^{n}\ \&\ (n)_{1}^{t}$, where $t$ $\in\ \mathbb{N}$ \& $t = a(n-a) + 1$\label{tabformula}}
\resizebox{\columnwidth}{!}{%
\begin{tabular}{|c|c|c|ccc|}
\cline{1-6}
\multirow{2}{*}{$(n)_{0}^{t}$}&\multirow{2}{*}{$(a)_{i=1}^{n}$}& \textbf{$(i)_{1}^t$} & \textbf{$1$} & \textbf{$to$} & \textbf{$t$}  \\\cline{3-6}

& & \textbf{\(S_i\)} & ${a \over 2}(a+1)$ & $to$ & ${a \over 2}(2n -a+1)$
\\
\cline{1-6}

\end{tabular}
}
\end{table}


\begin{figure}
     \centering
     \begin{subfigure}[b]{0.2\textwidth}
         \centering
         \includegraphics[width=\textwidth]{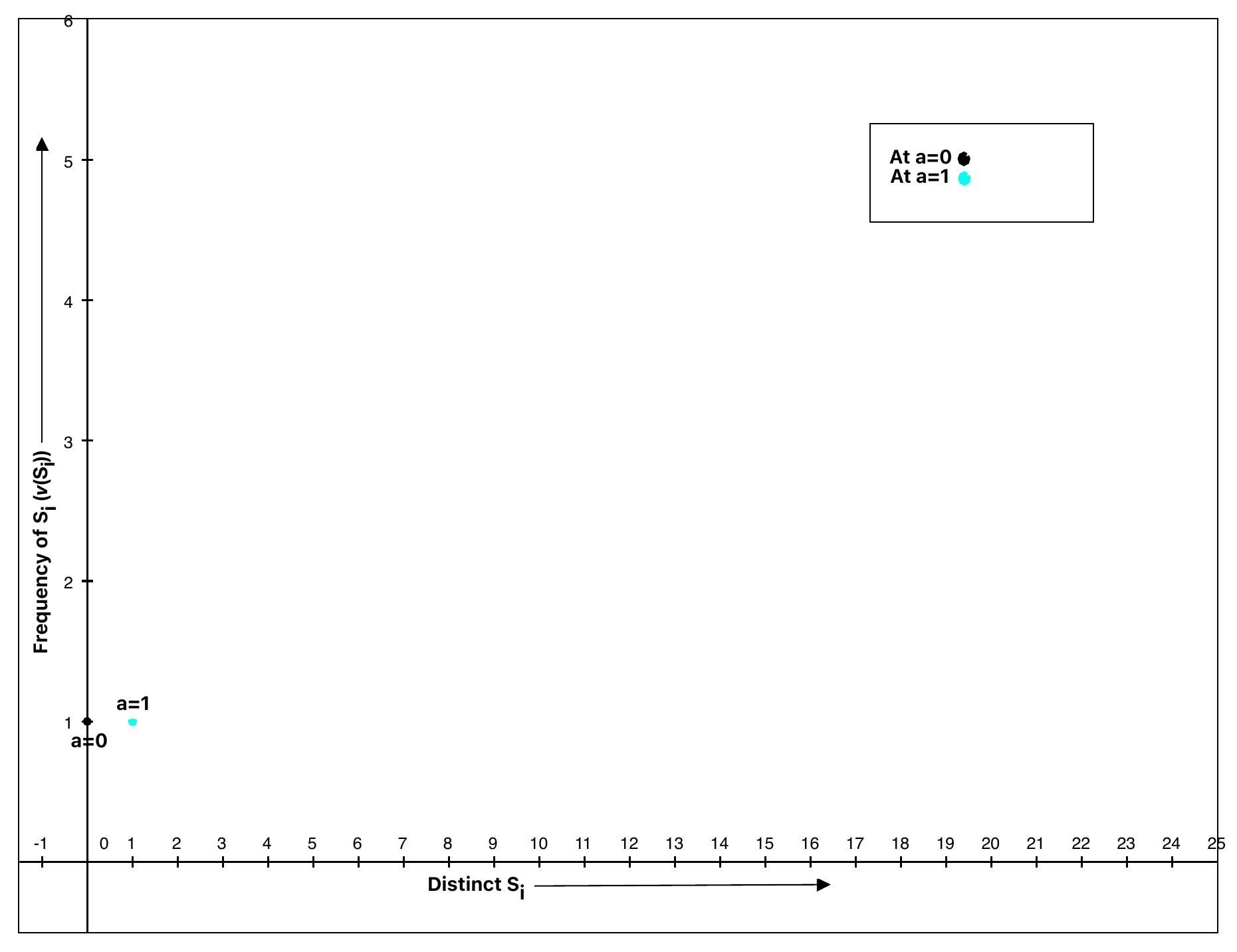}
         \caption{$n=1$}
         \label{n1}
     \end{subfigure}
     \hfill
     \begin{subfigure}[b]{0.2\textwidth}
         \centering
         \includegraphics[width=\textwidth]{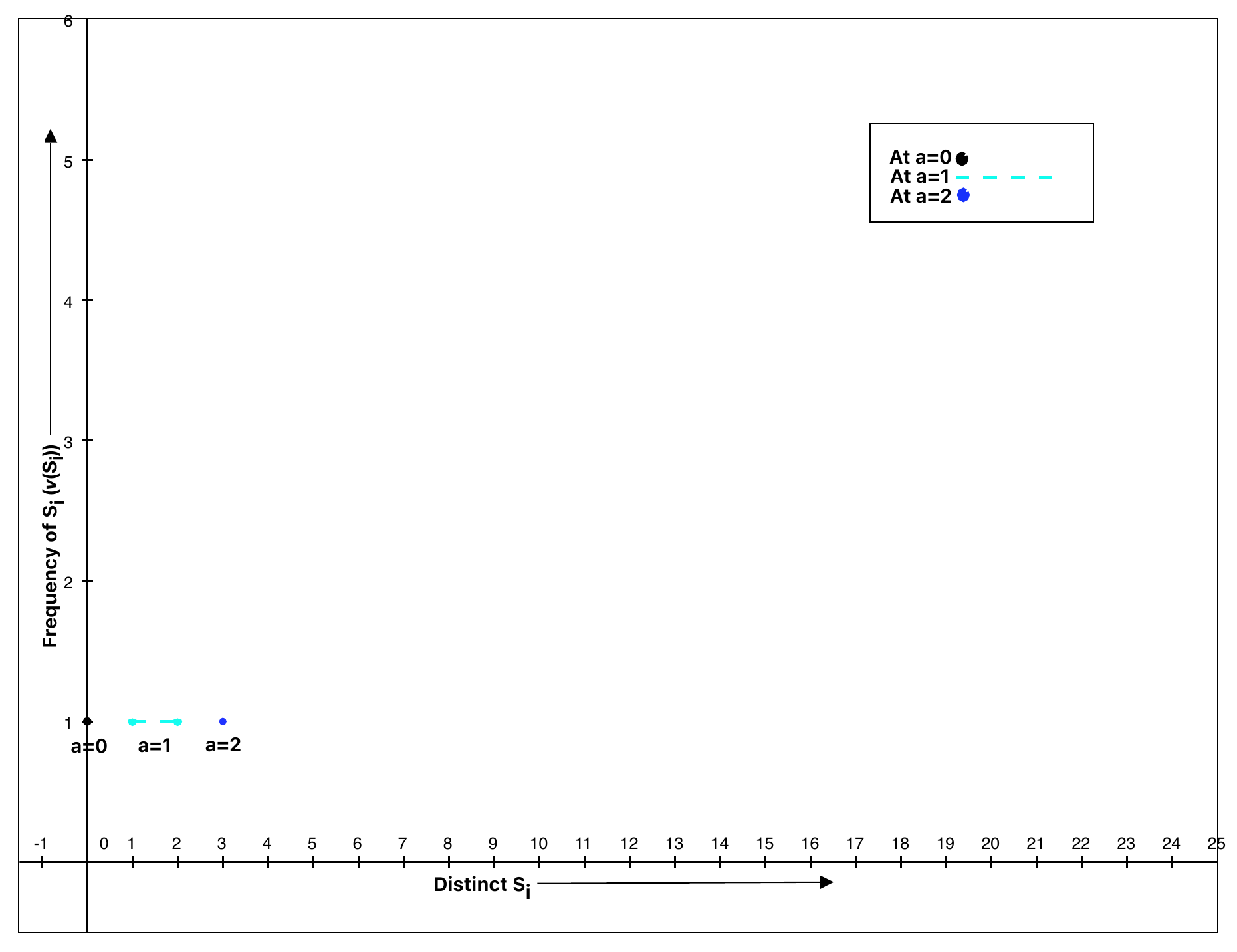}
         \caption{$n=2$}
         \label{n2}
     \end{subfigure}
     \hfill
     \begin{subfigure}[b]{0.2\textwidth}
         \centering
         \includegraphics[width=\textwidth]{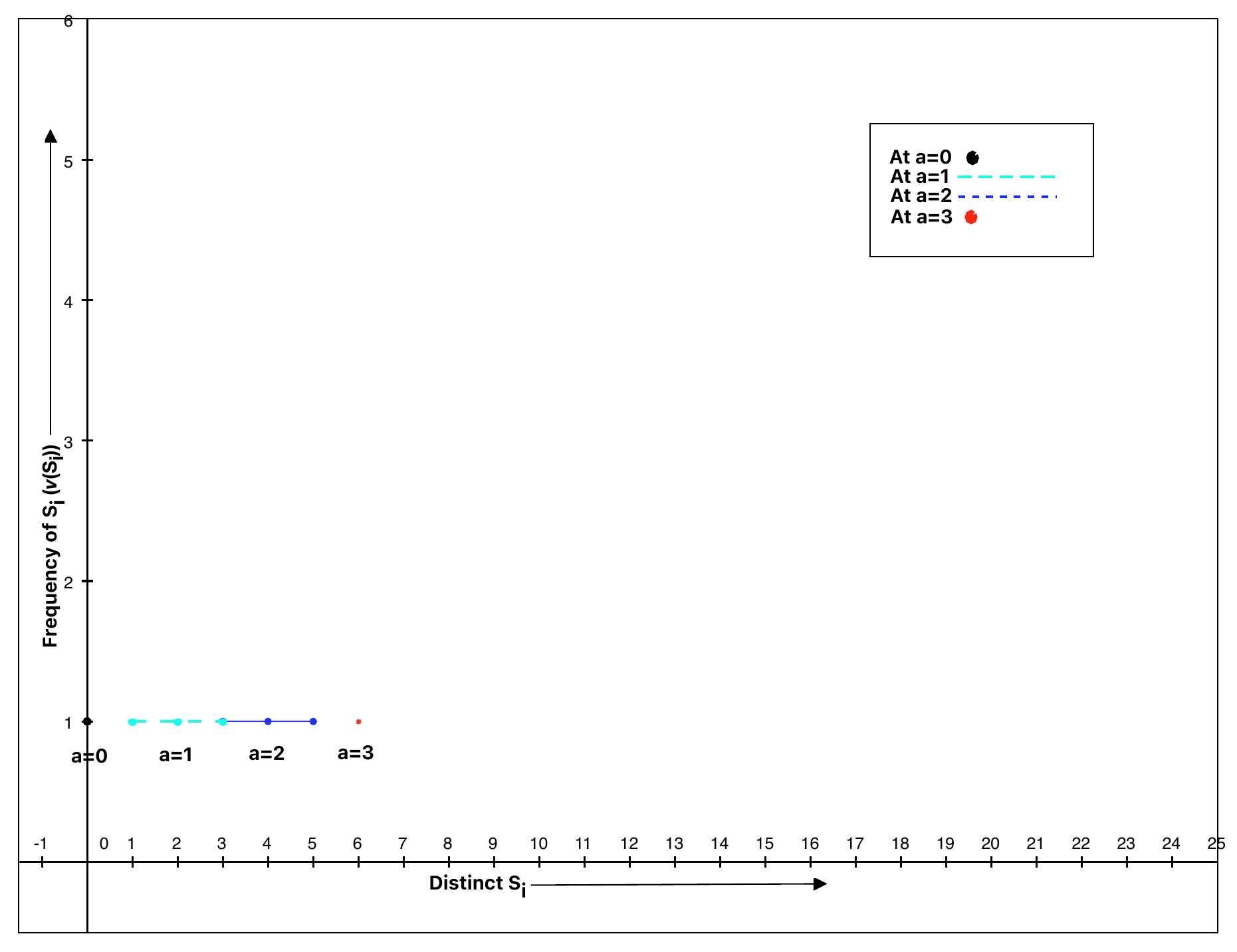}
         \caption{$n=3$}
         \label{n3}
     \end{subfigure}
     \hfill
     \begin{subfigure}[b]{0.2\textwidth}
         \centering
         \includegraphics[width=\textwidth]{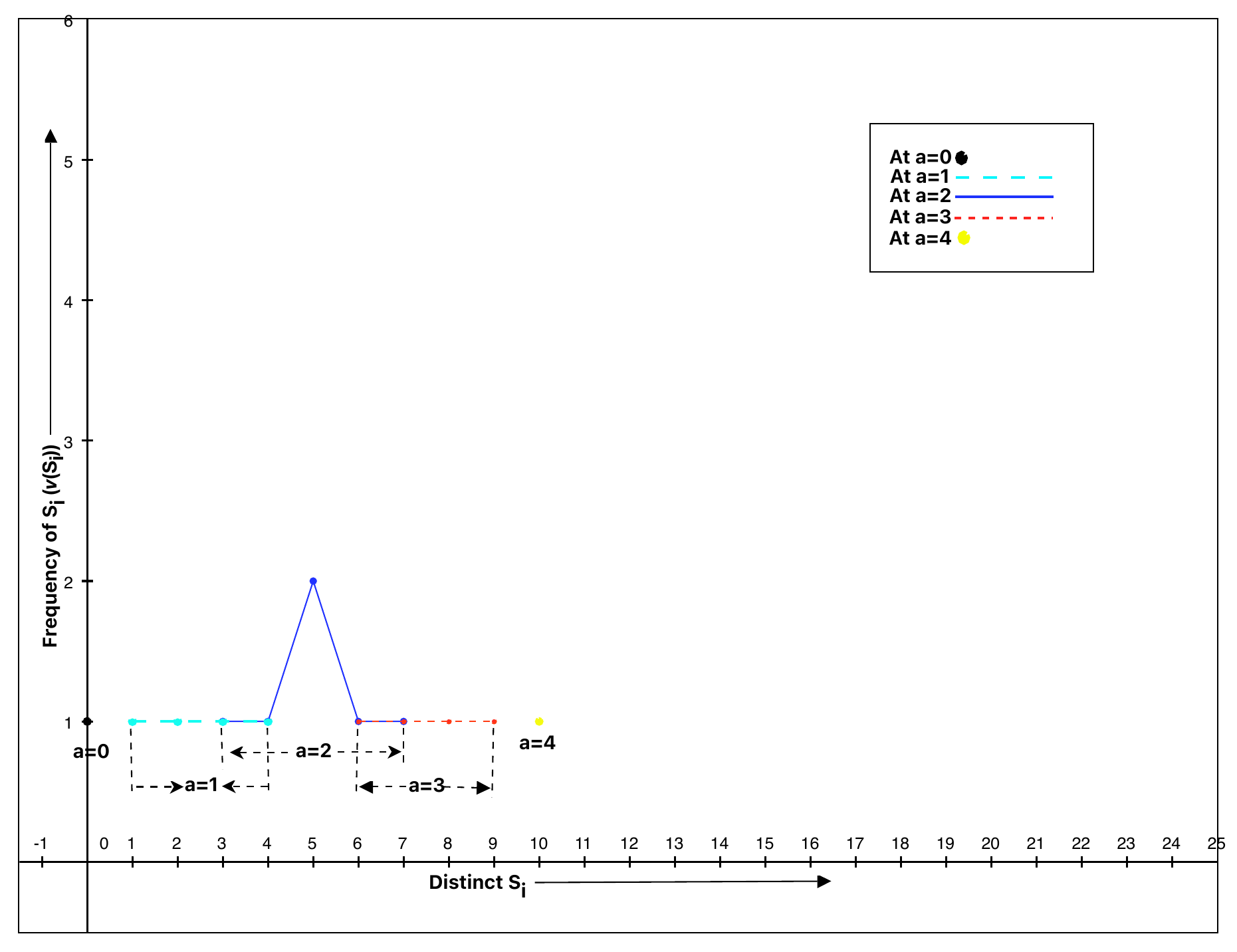}
         \caption{$n=4$}
         \label{n4}
     \end{subfigure}
     \hfill
     \begin{subfigure}[b]{0.2\textwidth}
         \centering
         \includegraphics[width=\textwidth]{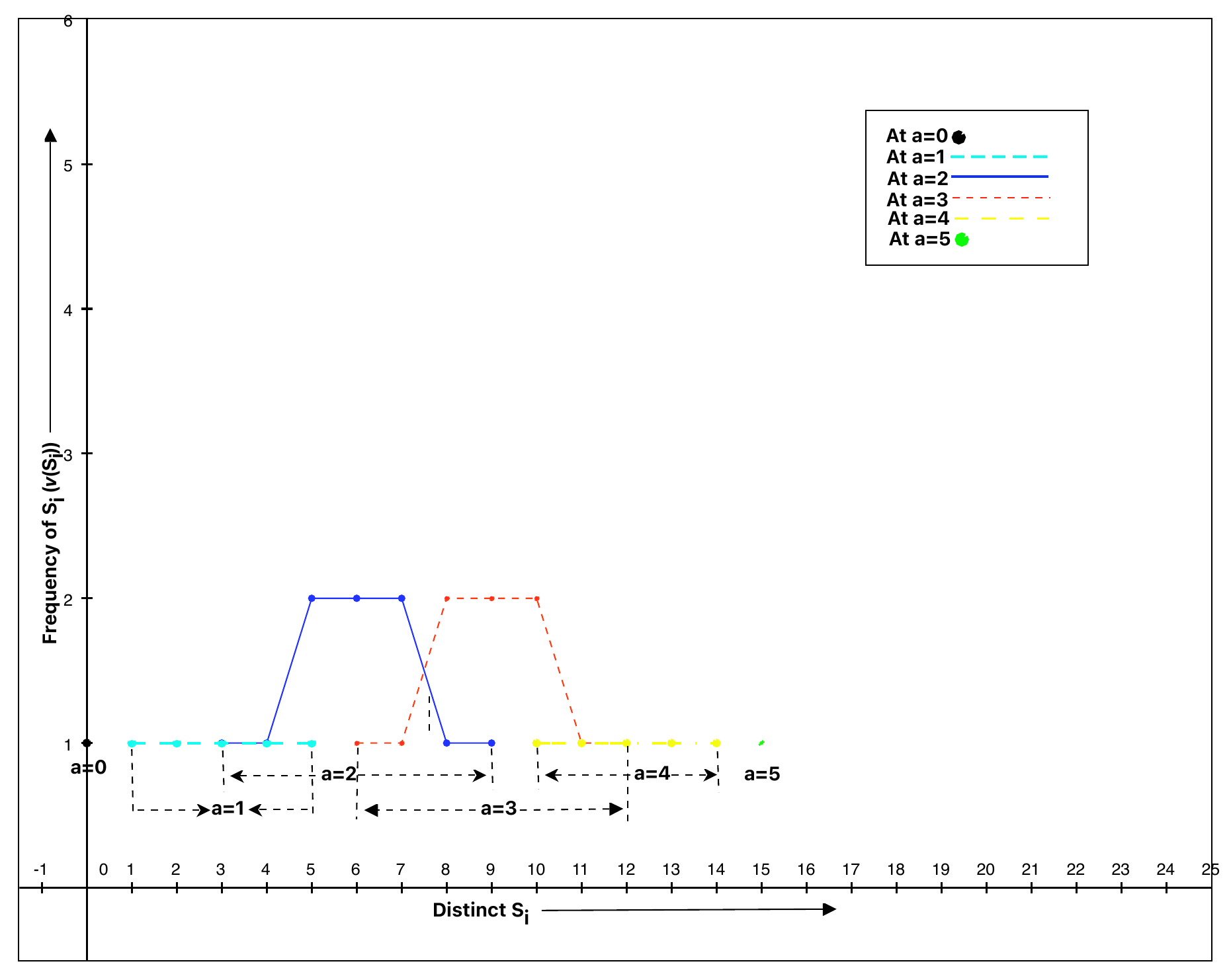}
         \caption{$n=5$}
         \label{n5}
     \end{subfigure}
     \hfill
     \begin{subfigure}[b]{0.2\textwidth}
         \centering
         \includegraphics[width=\textwidth]{sum3.pdf}
         \caption{$n=6$}
         \label{n6}
     \end{subfigure}
        \caption{Graph between frequency $\nu(S_i)$ \& distinct summations $(S_i)$ of sub-sequences at $(n)_{1}^{6}$ \& $(a)_{0}^{n}$, where $t$ $\in\ \mathbb{N}$ \& $t=a(n-a)+1$}
        \label{n1to6}
\end{figure}
 \begin{figure}[!t]
\centering
 \centerline{\includegraphics[width=1\columnwidth, height=0.8\columnwidth]{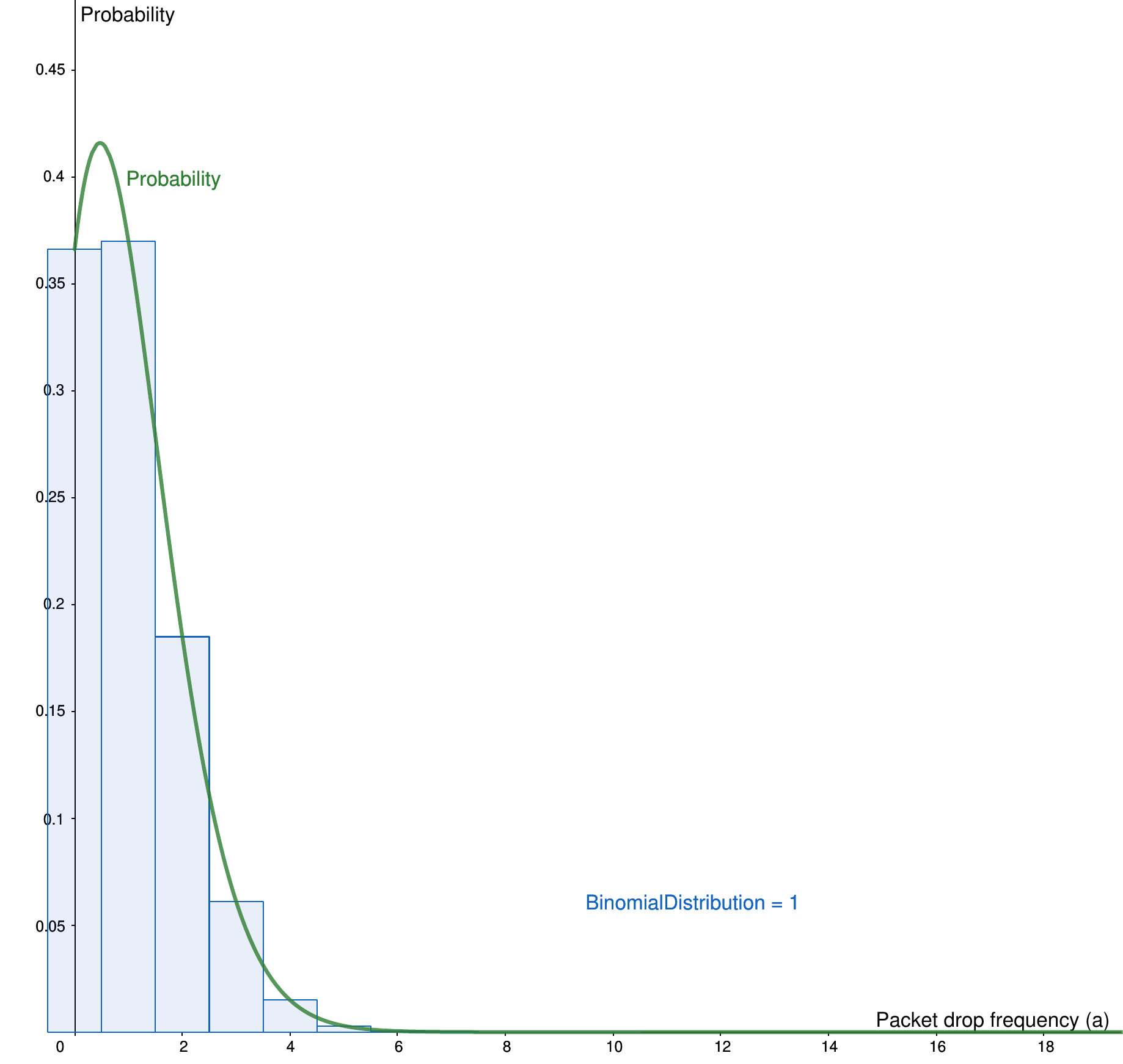}}
 \caption{ Success probability \& binomial distribution of packet drop frequency for n=100}
 \label{prob3}
 \end{figure}

In Figure \ref{fig:a23} illustrates a line graph between the frequency and the summation $S_i$ at $a=2$ \& $a=3$ for $n=6$ which is designed by using the data from Table \ref{tab6} \& \ref{tab4}. At $a=2$, we got a discrete line graph in positive $xy-Plane$ expressed as blue line. The line graph is symmetric and one side of the graph acts a image of the other which acts as pre-image by holding a mirror in the centre of symmetry. Similarly for the line graph of $a=3$ represented as red dotted line has a same feature of symmetry.

In Table \ref{taba06} we have calculated the summation $S_i$ and their frequencies for $a=0$ to $a=6$ at $n=6$.
Further extending the boundary of the analysis in Figure \ref{fig:a06} we described the line graph at $a=0$ to $a=6$ for $n=6$, and is designed by using the data from Table \ref{taba06}. In this analysis we can see that the frequency graph in increasing and reached to a shifting point from where they decrease in the same identical manner. The line graph for all $a = 0$ to $a=6$ designs a symmetrical frequency graph i.e. the line graph at $a=0$ to $a=2$ acts as a pre-image of $a=4$ to $a=6$ which also acts as an Image.
Both of the above illustration explains two conclusions:
\begin{enumerate}
\item The packet drop will be frequent from $S_i = 4\ to\ 17$.
\item Also at a given value of $a$ the packet drop will be frequent at the average value of its respective $S_i$.
\item The frequency of packet drop of $a =0$ and $a=n$ will be always equal such that their respective frequencies are always be $1$.
\end{enumerate}

In Table \ref{taba0n} we have enlisted the distinct summation and their frequencies for a given number of $n$ for each frequencies in range of $(a)_{1}^{n}.$

In Figure \ref{n1to6} depicts the data in Table \ref{taba0n} which illustrates the effect on frequency vs summation graph as the $n$ increases. In Figure \ref{n1}, \ref{n2}, \ref{n3}, \ref{n4}, \ref{n5}, \& \ref{n6} the graph is symmetric and the frequency vs summation graph has similar rate of change of frequency with respect to summation. Further the amplitude of the graph increases with the increase in the $n$ and follows same graph curve and symmetry. Additionally, for any value of $n$ the average summation values have highest frequency while the minimum and the maximum summation value are always constant to $1$. The same relation is followed by the individual graph for a given value of $a$.

In Table \ref{tabformula} we have formulated a general method of calculating the summation $S_i$ of the combinations formed from a set $B$. i.e.
\begin{align}
S_1 = {a \over 2}(a+1)\\
S_{t} = {a \over 2}(2n-a+1)
\end{align}
Where $t=a(n-a)+1$ which is the last distinct partial sum.
The distinct summation $S_i$ of the sub-sequences $B'$ of size $a$ of sequence $B$ of size $n$ ranges from $[{a \over 2}(a+1),{a \over 2}(2n-a+1) ]$.

 \begin{algorithm}[t]
\DontPrintSemicolon
\caption{Identifying the minimum, average and maximum time wastage by Theorem \ref{theorem_3} i.e $a \le n$ in the set of position of node $\{n_1,n_2,n_3 \cdots, n_n\}$. \label{Algouu}}
\SetKwInOut{KwIn}{Input}
\SetKwInOut{KwOut}{Output}
\KwIn{A Set $[n_i]$, $i=1, 2, \cdots, n$, where each element is an integer, a $\le$ n.}
\KwOut{Max, Min and Avg Time Wastage.}
\textbf{function} \textsc{TimeWastage}($n_1,n_2,n_3 \cdots, n_n$)\\
\textbf{parameter:} $a$\\
Divide the Set $\{n_1,n_2,n_3 \dots, n_n\}$ into $s = [{^nC_a}]$ subsets $\{n_1,n_2,n_3 \cdots, n_a\}$ of length $a$ without repetition.\\
$s$ number of combination with size $a$ from the set $\{n_1,n_2,n_3 \dots, n_n\}$\\
\For{$i\gets1$ \KwTo $s$}{
$Sum_i \gets$ Add all items in the subset.\\
$T_i \gets Sum_i({^p}T_d + {^p}T'_d) + aT_f.$\\
\If{$Sum_i = \frac{a}{2}(2n-a+1)$}
{\KwRet{$Max(T_i);$}\\}
\If{$Sum_i = \frac{a(a+1)}{2}$}{\KwRet{$Avg(T_i);$}\\}
\If{$Sum_i = \frac{a(n+1)}{2}$}{\KwRet{$Min(T_i);$}}
}
\end{algorithm}

\subsection{Using Bernoulli's Trails}
Coming back to the the network topology, the frequency of packet drop $a \le n$ has range of $(a)_{1}^{n}$. To find out what would be the probability that in a complete transmission of a single packet what would be the frequency of drop $a$ i.e. either the packet will be dropped twice, thrice, 10 times, 20 times to $n-times$. By finding this will help us to give an inside of the likely of how frequency of packet drop will happen on a specific network topology.

\begin{table}[!t]
\caption{The relative minimum, average and maximum time wastage at specific frequency $a$ for given value of $n$.\label{tab1}}
\resizebox{\columnwidth}{!}{%
\begin{tabular}{c|c|c|c}
\hline
{\centering Frequency $(a)$} & {\centering Min $(T_w)$} & {\centering Avg $(T_w)$} & {\centering Max $(T_w)$}\\ [3ex]
\hline
1 & 3 & 14 & 25\\
2 & 6 & 36 & 50\\
3 & 9 & 54 & 75\\
4 & 12 & 72 & 100\\
5 & 15 & 90 & 125\\
6 & 18 & 108 & 150\\
7 & 21 & 110 & 175\\
8 & 24 & 126 & 200\\
9 & 27 & 144 & 225\\
10 & 30 & 162 & 250\\
11 & 33 & 180 & 275\\
12 & 36 & 198 & 300\\[1ex] 
\hline
\end{tabular}
}
\end{table}

For this special reason we have find a special theorem i.e. Bernoulli's Theorem on Trails. In Bernoulli's trail/theorem a $n-trials$ are applied on $n-objects$ with equally likely probability of $1/n$ getting a desired number $x\in n$, then the probability of getting a k-success for the $n$ Bernoulli's trails $B(n,k)$ on the number $x$ is given by:
\begin{align}
   P(k) = \binom{n}{k}\left(\frac{1}{n}\right)^{k}\left(1-\frac{1}{n}\right)^{n-k} 
\end{align}
While in our case we have the same situation of the Bernoulli's trails but in different angle. In our case the desired number i.e. the number $x$ which acts as success will be choosing the suitable $k-trial$ in $n-trails$ in one iteration of $x$ on $n-trails$. Hence $n-nodes$ acts as $n-trails$ and the packet acts as the desired number $x$ and the $a-times$ packet drop defines $k-times$ success. Then the Bernoulli theorem can be applied to find the $a-times$ success rate i.e.  for the Bernoulli's trail $B(n,a)$ the $a-times$ success probability is given by:
\begin{align}
P(a) = \binom{n}{a}\left(\frac{1}{n}\right)^{a}\left(1-\frac{1}{n}\right)^{n-a}
\end{align}
Let the n-trails is given $B={1,2,3,4 \dots n}$ \& the frequency of drop $a\le n$ then the a-times success probability $\forall\ a\ \in\ [1,n]$ is given as:

\begin{align}
P(0) &= \frac{(n-1)^{n}}{0!}\left(\frac{1}{n}\right)^{n}\\
P(1) &=\frac{(n-1)^{n-1}}{1!}\left(\frac{1}{n}\right)^{n-1}\\
P(2) &= \frac{(n-1)^{n-1}}{2!}\left(\frac{1}{n}\right)^{n-1}\\
P(3) &= \frac{(n-2)(n-1)^{n-2}}{3!}\left(\frac{1}{n}\right)^{n-1}\\
P(4) &= \frac{(n-3)(n-2)(n-1)^{n-3}}{4!}\left(\frac{1}{n}\right)^{n-1}\\
P(5) &= \frac{(n-4)(n-3)(n-2)(n-1)^{n-4}}{5!}\left(\frac{1}{n}\right)^{n-1}\\ \notag
\cdots &\cdots\\\notag
\cdots &\cdots\\\notag
\cdots &\cdots\\
P(n-5) &= \frac{(n-4)(n-3)(n-2)(n-1)^{6}}{5!}\left(\frac{1}{n}\right)^{n-1}\\
P(n-4) &= \frac{(n-3)(n-2)(n-1)^{5}}{4!}\left(\frac{1}{n}\right)^{n-1}\\
P(n-3) &= \frac{(n-2)(n-1)^{4}}{3!}\left(\frac{1}{n}\right)^{n-1}\\
P(n-2) &= \frac{(n-1)^{3}}{2!}\left(\frac{1}{n}\right)^{n-1}\\
P(n-1) &= \frac{(n-1)}{1!}\left(\frac{1}{n}\right)^{n-1}\\
P(n) &= \frac{1}{0!}\left(\frac{1}{n}\right)^{n}
\end{align}

  \begin{figure}[!t]
 \centering
 \centerline{\includegraphics[width=1\columnwidth, height=0.8\columnwidth]{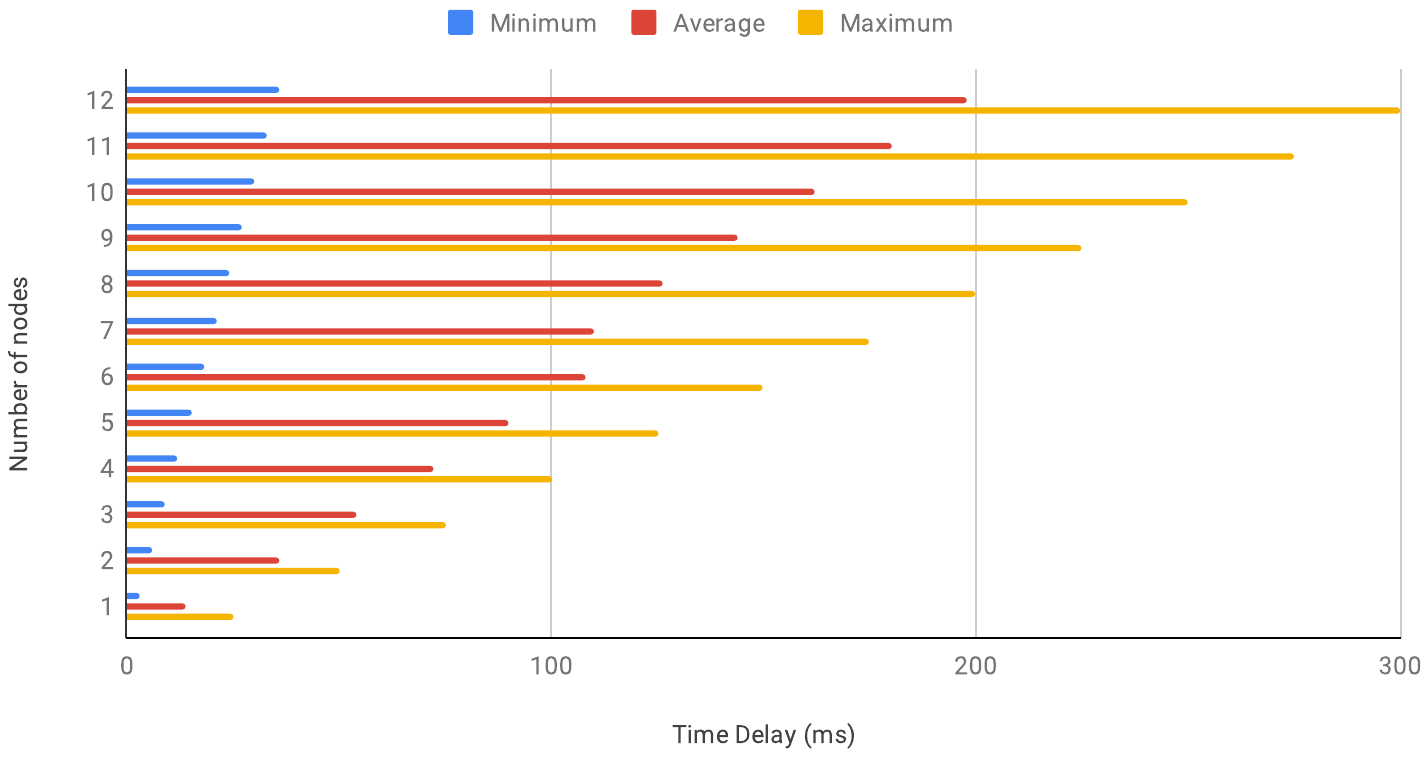}}
 \caption{Relative Maximum, minimum and average time delay for all values of frequencies $a\le n$ at given value of $n=12$ as per Table \ref{tab1}.}
 \label{net_topo}
 \end{figure}
 
In Figure \ref{prob3} is a illustration of success probability of occurrence of specific value of a for a given value of n, here $n=300$ i.e. number of nodes. The probability graph makes a y-intercept of $~0.35$ at $a=0$ then the probability graph abruptly increases until $a=1$ and then it gradually decreases at the same point till the probability reaches to approx value of zero at $~a=5$ and then it keeps this constant value upto $a=300$ which is the highest simulation value at highest frequency range for $n=300$.
 
Further more in Figure \ref{prob3} we have applied the binomial distribution as shown in blue bar/strips on the success probability vs frequency. The same behaviours is shown as by the probability graph. 
 
The probability graph in Figure \ref{prob3} for the frequency of packet drop identifies the degree of chances of respective frequency i.e. $a$. The simulation is carried-out in Geo-Graph (Mac Mathematical Simulator) using the Bernoulli's equation of success probability at a for n. From the simulation experiment shown in Figure \ref{prob3} that the probability of occurrence of packet drop frequency is maximum for $a =0\ to\ 5$ for $n=300$ that is about approx $2\%$ of $a=300$. Hence, for any given value of $n$ the probability of occurrence of $a$ is maximum for the first $2\%$ of the maximum frequency i.e. $a=n$ $\forall\ n\ \in\ \mathbb{N^+}$.
 
Then the probability of chance of occurrence of maximum, minimum and average $S_i$ for a constant value of $(n,a)$ is given as:
\begin{align}
P(\max(S_i)) = P(a)\frac{\nu(max(S_i))}{\sum_{i=1}^{t}\nu(S_i)}\\
P(\min(S_i)) = P(a)\frac{\nu(min(S_i))}{\sum_{i=1}^{t}\nu(S_i)}\\
P(avg(S_i)) = P(a)\frac{max(\nu(S_i))}{\sum_{i=1}^{t}\nu(S_i)}\\ \notag
where\ t = a(n-a)+1.
\end{align}
Since the value of $\nu(\max(S_i))$ \& $\nu(\min(S_i))$ is always equals to $1$ for any value of (n,a) from Table 6. also $\sum_{i=1}^{t}\nu(S_i) = \binom{n}{a}$ and $\max(\nu(S_i)) = nD(S_i) - n = (a(n-a)+1) - n$. While the value of $\max(\nu(S_i))$ for odd values of $nD(S_i)$ is given as:
\[max(\nu(S_i)) = nD(S_i) - n\]
Then the probability of occurrence of maximum, minimum and average $S_i$ is given by:
\begin{align}
P(\max(S_i)) &=\left(\frac{1}{n}\right)^{a}\left(\frac{n-1}{n}\right)^{n-a}\\
P(\min(S_i))
&=\left(\frac{1}{n}\right)^{a}\left(\frac{n-1}{n}\right)^{n-a}
\end{align}
The probability of occurrence of average $S_i$ for odd number of $nD(S_i)$ is given by:
\begin{align}
P(avg(S_i))
=\left(\frac{1}{n}\right)^{a}\left(\frac{n-1}{n}\right)^{n-a}(a(n-a)+1-n)
\end{align}
In Figure \ref{net_topo} illustrates the graphical representation of data in Table \ref{tab1} which illustrates the bar graph of possible relative maximum, minimum and average time wastage at a given packet drop frequency i.e. $a$ at a given value of $n=300$. The blue bar in Figure \ref{net_topo} defines the minimum time wastage at a given frequency i.e. $a$. Similarly, the red bar and yellow bar defines the average time wastage and maximum time wastage respectively at a given frequency i.e. $a$. These value in Table \ref{tab1} are calculated by using the Equation \ref{max}, \ref{min} \& \ref{avgT} and using the pre-defined parameter values as: ${^p}T_d = 1ms$, ${^p}T'_d = 0.5ms$, $T_f = 1ms$, in the same equations for each value of $a$ for a given value of $n$. In Algorithm \ref{Algouu} explains the method of calculating maximum, minimum and average time wastage.
\section{Conclusions \& Future Work}
\label{sec:conc}
We concluded that the time wastage increases with the increase in the number of re-transmissions for a single packet. We adopted new theorems and corollaries to clearly calculate the time wastage resulted due to continuous use of Path MTU Discovery in IPv4 and IPv6 networks. The comprehensive analysis carried-out in the paper are one of the first stages of research in time delays encountered in Path MTU Discovery in IPv4 and IPv6 networks.

In this paper, we had try to show the effect of Path MTU discovery in IPv4 \& IPv6 network using mathematical, logical and graphical representation \& analysis. The packet drop frequency in network follows a order of k-combination with $\binom{n}{a}$ number of combination for a particular frequency $a$ at given value of $n$ where $a\le n$. We further concluded that the graph between the frequency of packet drop to the summation of respective k-combination is symmetric in nature. 

Further, we are able to represent the path MTU discovery in a structured manner by giving different insight on mathematical and statistical structures which we have represented using graphical representation.

We depicted the probability of calculating the relative minimum, average and maximum time wastage for any frequency of packet drop at given value of $n$. Further using the same we were able to calculate mathematically and graphically the relative minimum, average and maximum time delay. We were nearly able to explain how much relative difference between minimum, average and maximum time wastage for any frequency of packet drop using graphical representation. This representation help in designing the network engineers and other researchers to place the nodes in the network in order to decrease the packet drop rate and the time delay.

Further, we explained through using Bernoulli's theorem and the binomial distribution the success probability of the packet drop frequency $a$ $\forall a\le n$ which shows that the probability is higher for packet drop rate for beginning $2\%$ of the total nodes in the path. 

We also concluded that for a specific packet drop frequency at a given value of $n$ the frequency of partial summation is maximum at the average value of $S_i$ and hence the packet drop rate is higher, while at the $i=0$ and $i=a(n-a)+1$ the frequency of $S_i$ is constant and is equal to one and hence the packet drop rate is minimum.

While the value of $S_i$ increases on moving from $a=1$ to $a=n$ and hence $S_i$ value acts as a coefficient of $(^{p}T_d + ^{p}T'_d)$ resulting the overall increase in the value of $T_{w}$.

The packet drop frequency in the network for a specific number of nodes in a given path follows the $k-combinations$ of a arithmetic sequence with $distance\ =\ 1$ and first term as $'1'$. Also the time wastage increases with the increase in the node position. The time wastage is higher when the packet drop near the destination node and lower when the packet is dropped near the source node while the time wastage increases on dropping packet from source towards destination.

Further conclusion that is drawn from these analysis is that the time wastage due to Path MTU discovery has an asymptotic lower bound  of \(\Omega(n)\) and upper bound of \(\Theta(n^2)\) with the domain depending only on the number of nodes the packet traverses the path.

The analysis can be used by the research community to fine tune some of the parameters in order to reduce the delays associated with PMTUD protocol. The analysis can also been used as the way to find the optimistic and robustness of a new protocol design than the present one.

\ifCLASSOPTIONcompsoc
  \section*{Acknowledgments}
\else
  \section*{Acknowledgment}
\fi

This research was carried-out \& supported by Janibul Bashir's Laboratory, National Institute of Technology, Srinagar, Jammu \& Kashmir, India.

\ifCLASSOPTIONcaptionsoff
  \newpage
\fi

\bibliographystyle{IEEEtran}
\bibliography{IEEEtran}

\begin{IEEEbiography}[{\includegraphics[width=1in,height=1.25in,clip,keepaspectratio]{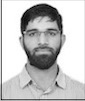}}]{Janibul Bashir}
is currently working as an Assistant Professor in the Department of Information Technology at National Institute of Technology, Srinagar, Jammu \& Kashmir, India. He received Doctoral of Philosophy in Computer Science and Engineering from Indian Institute of Technology Delhi in 2020 and Master's of Technology in Computer Science Engineering from Indian Institute of Technology (Indian School of Mines), Dhanbad in 2017. He earned his B.TECH degree in Information Technology from National Institute of Technology, Srinagar in 2014. Before joining NIT Srinagar, he has worked as Software Engineer at Samsung India Software Operations, Bangalore and was awarded with Spot award by Samsung Company. He currently directs the GAASH research group at NIT Srinagar. His research work is on improving the performance of multi-core systems, on-chip security, optical network, IPv6 network, application of machine learning techniques in the computer architecture domain (emerging technologies, network-on-chip, thermal management). He is author \& co-author of more than 20 journals which are published in premier journals. He has extended interested in the operating systems and parallel programming (distributed systems). He can be reached at janibbashir@nitsri.ac.in.
\end{IEEEbiography}

\begin{IEEEbiography}[{\includegraphics[width=1in,height=1.25in,clip,keepaspectratio]{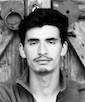}}]{Ishfaq Hussain}
received his Bachelors of Technology in Information Technology from National Institute of Technology, Srinagar, India in 2018. He was working as a Visiting Researcher at Janibul Bashir’s Laboratory at National Institute of Technology, Srinagar, Jammu \& Kashmir, India during this research study. Currently he is working as an independent researcher and reviews pre-published papers of journals like Oriental Journal of Computer Science and Technology and similar computer science journals. His research interest includes protocol designing, new generation networks, mathematics of randomness, stochastic methods, deep learning, optimisation \& algorithm designing. He can be reached at ishfaqhussain90@gmail.com.
\end{IEEEbiography}
\vfill

\end{document}